\documentclass[10pt,twocolumn]{article}
\usepackage{fullpage}
\usepackage[margin=.8in]{geometry}
\usepackage{balance}
\usepackage{graphicx}
\usepackage{graphicx}
\usepackage{amsthm}
\usepackage{amsfonts}
\usepackage{amsmath}
\usepackage{subfigure}
\usepackage{xcolor}
\usepackage{qtree}
\usepackage{algorithm}
\usepackage{algorithmic}
\usepackage{epsfig}
\usepackage{url}
\usepackage{thmtools}
\usepackage{appendix}

\declaretheorem{definition}

\declaretheorem{lemma}
\newcommand{\para}[1]{{\vspace{4pt} \bf \noindent #1 \hspace{10pt}}}
\newcommand{\fixme}[1]{\textcolor{black}{#1}}

\makeatletter
\newcommand*{\defeq}{\mathrel{\rlap{%
                     \raisebox{0.3ex}{$\m@th\cdot$}}%
                     \raisebox{-0.3ex}{$\m@th\cdot$}}%
                     =}
\makeatother

\graphicspath {{./} {figures/}}
\begin{document}

\title{A Geometric Distance Oracle for Large Real-World Graphs}
\author{Deepak Ajwani\thanks{Bell Labs, Alcatel-Lucent, Dublin, Ireland},
W. Sean Kennedy\thanks{Bell Labs, Alcatel-Lucent, Murray Hill, NJ, USA},
Alessandra Sala$^{*}$, Iraj Saniee$^{\dagger}$
}
%%\affil{\texttt{\{kennedy,iis\}@research.bell-labs.com}, \texttt{onarayan@ucsc.edu}}

%\date{}
%%\author{Deepak Ajwani\ \and W. Sean Kennedy\ \\
%%\and Alessandra Sala\ \and Iraj Saniee\}

%\institute{}

\maketitle

\begin{abstract}

Many graph processing algorithms require
determination of shortest-path distances between arbitrary numbers of node pairs. Since computation of exact distances between 
all node-pairs of a large graph, e.g., 10M nodes and up, is 
prohibitively expensive both
in computational time and storage space, 
distance approximation is often used in place of
exact computation.  
A {\em distance oracle} is a data structure 
that answers inter-point 
distance queries more efficiently than in
standard $O(n^2)$ in time or storage space for an $n$ node
graph, e.g., in $O(n \log{n})$.
In this paper, we present a novel and scalable distance 
oracle that leverages the hyperbolic core of
% intrinsic
% (i.e., without embedding) hyperbolicity of 
real-world large graphs for fast and scalable
distance approximation via spanning trees.  We show 
empirically 
that the proposed oracle significantly outperforms 
prior oracles  
on a random set of test cases
drawn from public domain graph libraries.  
There are two sets of prior work against which
% that 
% use geometric
% properties of graphs for an oracle design and against 
% which 
we benchmark our approach. The first
set, which often outperforms all other oracles, 
employs embedding of the graph into 
low dimensional Euclidean spaces with carefully 
constructed hyperbolic distances, but provides no 
guarantees on the distance estimation error. 
The second set leverages
Gromov-type tree contraction of the graph with the
additive error guaranteed not to exceed $2\delta\log{n}$, 
where $\delta$ is the hyperbolic constant of the graph.  
We show that our proposed oracle 1) is significantly 
faster than those oracles that use hyperbolic embedding 
(first set) with similar approximation error
and, perhaps surprisingly, 2) exhibits substantially lower 
average estimation error compared to Gromov-like
tree contractions (second set). 
We substantiate our claims through numerical computations 
on a collection of a dozen real world networks and 
synthetic test cases
from multiple domains, ranging in size from 10s of 
thousand to 10s of millions 
of nodes.

\end{abstract}

%\newpage
%\vspace{-0.2cm}
\section{Introduction}
\label{sec:intro}
The explosion of available information in the past decade, in part due to 
the rapid shift towards online media and interactions has led many research, business and marketing communities to store and analyze very large
data sets.
Mining these data sets promises to reveal a wealth of information about 
the interests of and the kind of interactions between subscribers, groups, 
people, objects and even ideas.
These interactions are often naturally represented by graphs, where for example, nodes
correspond to the entities of interest and the (weighted) links represent the strength of the interaction between them.
Graphs extracted for data mining are often massive, comprising of 
millions to billions of connections.
At this scale, graph algorithms requiring $\Omega(n^2)$ computational
steps or storage reach their useful limit in terms of 
run time and memory requirements.
There is clearly a need for implementations 
of graph computational primitives at this scale.

Computing shortest path distance between arbitrary nodes of a graph is 
among such fundamental computational primitives.
Many data mining schemes invoke this computational primitive in the scale 
of the number of nodes, and therefore it is imperative that this 
computation can be carried out very rapidly and with limited
memory consumption.
Distance oracles are among the many approaches that have been proposed and 
used for simplification of shortest distance computation for
large-scale graphs.  A distance oracle involves an auxiliary data structure which
is cheaper to compute and fast to query.  It should ideally satisfy the 
following four properties:
\begin{enumerate}
\item (Initial Processing Speed) the computation involved in the creation of 
  the auxiliary data structure should be scalable, e.g. be $O(n)$ or $O(n\log{n})$ 
  (and not $O(n^2)$ or more complex),
\item (Storage) the auxiliary data structure should be represented in much 
  smaller space (storage or memory) compared to storing shortest path 
lengths between all node pairs,
\item (Fidelity) path length (estimation) queries using the auxiliary 
  data-structure should return 
  distances which are as close as possible (if not equal) to the actual 
  distances,
\item (Query Speed) the time required to query the distance between any two 
  nodes should be very small (e.g., small fraction of a second). 

\end{enumerate}

In this paper, we focus on a distance oracle for 
the large-
scale graphs that 
models interactions arising naturally in the real-world, as 
in online social networks,
call graph, co-authorship, citations, hyperlinks in world 
wide
web and similar graphs. We refer to these graphs as real-world graphs.  
In recent years, considerable effort has been expended in 
developing approximate distance oracles on these graphs, 
e.g., see~\cite{rigel,Tretyakov11,Qiao11,orion,20BFS,akiba13,qi13}, 
however these heuristics lack a theoretical foundation that would
explain their observed accuracy in practice. 
On the other hand, there are theoretical
results
~\cite{gav-ly,CD00,chepoi,chen12,microsoft,gromov, Frenchies} 
which provide guaranteed approximation bounds for 
specific graph classes.
However, the accuracy of many of these methods has not been empirically evaluated on real-world graphs.

\para{Our Contribution.} We present a novel
distance approximation oracle that leverages 
the notion of graph hyperbolicity~\cite{gromov} observed 
in real-world 
networks~\cite{ns11,montgolfier11,kennedy13,inria13}
and specifically uses
the `hyperbolic core' of the graph ~\cite{ns11,jonk11}
for a spanning tree approximation.
Hyperbolicity captures the geometric notion of 
negative curvature in smooth geometry, which we formally 
define in Section~\ref{sec:hyperbolicity} 
in the context of a graph; intuitively, and crucially,
as observed recently~\cite{ns11,jonk11},
it expresses the case that a fixed fraction 
$\Theta(n^2)$ of all shortest paths 
traverse a small set of nodes in the graph, thus
relative to this small set of nodes, the graph is 
tree-like in
some fundamental ways, a property which we exploit.
Our approach also bridges the gap between a) the 
theoretical understanding of tree approximations 
for hyperbolic graphs~\cite{gromov,gav-ly,
CD00,chepoi,chen12,microsoft,Frenchies} 
and b) the recent practical distance approximation
solutions~\cite{rigel,Tretyakov11,Qiao11,orion,
20BFS,akiba13,qi13} which exhibit high accuracy.

Our approach is distinct from both sets of prior work 
since we use neither 
1) any kind of hyperbolic embedding~\cite{rigel,orion}, nor 
2) any form of Gromov-type tree contraction and
labeling schemes~\cite{CD00,gav-ly,chepoi}.  
We construct a breadth-first search 
spanning tree prioritized on the nodal betweenness 
centrality and its approximation by nodal degree, 
as detailed in Section~\ref{BFSapproach}.
The height of these trees is almost always
$O(\log{n})$\footnote{This is due to the well-cited 
small world property of real-world graphs where 
the diameter is 
observed to scale like $O(\log{n})$ or even faster
(unless the graph has a dominant line in it).} 
and we use this fact to 
encode distances in trees with $O(n \log{n})$ bits 
(or $O(n)$ words) to support $O(\log{n})$ 
distance queries.
%
%engineer a data 
%structure that encodes
%distances in trees with $O(n \log{n})$ bits 
%(or $O(n)$ words) and small 
%constant factors to support $O(1)$ distance queries. 
Thus, our
proposed oracle 1) creates tree approximations
rapidly, 2) uses $O(n)$ words of space for storage for an 
$n$-node graph, 
3) has high fidelity and
low distortion (due to the tree root being in the
hyperbolic core) and 
4) returns queries very rapidly, meeting 
the key criteria we set out for an effective oracle. 
We also demonstrate empirically that
the accuracy of the proposed oracle on large real-world 
graphs is significantly better than the theoretical bounds
and competitive with the best practical solutions.

\para{Empirical and Synthetic Datasets.} Our empirical comparison is
carried out on a wide range of 
real-world and synthetic graphs, representing a variety of different
topological structures 
from the irregular connectivity of {\em small world} graphs to the
geometric symmetry and regularity of grid and {\em hypergrid}
graphs. Specifically, we experiment on two online social graphs from
Facebook, i.e., Santa Barbara and New York
(from~\cite{ucsb_facebook_data}), two call graphs from
anonymous telecom operators, two collaboration
networks~\cite{Lasagna}, a Google news graph, a Peer-to-Peer network,
a web graph,  and two synthetic networks, that we call
FlatGrid and HyperGrid. FlatGrid is a square lattice.
HyperGrid is a bona fide hyperbolic locally planar 
graph with degree 7 for interior nodes, lower degrees for
boundary nodes, and triangular faces. 
In Table~\ref{table:graphstats1} we summarize 
the topological and geometric characteristics of 
these datasets. 

\para{Outline of Paper}  Section~\ref{sec:rel_work} summarizes the
relevant prior work on distance oracles. In
Section~\ref{sec:hyperbolicity}, we describe the notion of
graph hyperbolicity that is at the core of our proposed geometric oracle. Next, 
in Section~\ref{sec:oracle},
we describe the proposed distance approximation oracle 
and also present various alternative
oracles for benchmarking. 
Subsequently in Section~\ref{sec:results}, we describe 
the experimental
methodology and summarize our results.
%In Appendix~\ref{sec:tree_encode} 
%we describe the tree distance encoding
%scheme in the proposed oracle. 

\begin{table}[tr]
\begin{center}
\caption{\small $9$ Real and $2$ Synthetic Benchmark Graphs}
\vspace{0.1in}
\begin{footnotesize}
\begin{tabular}{|c|c|c|c|c|}
\hline
%  && \multicolumn{2}{|c|}{Scalar  values} \\ \cline{3-4}
   %& Statistical &\multicolumn{7}{|c|}{Distribution Models} \\ \cline{3-9}
   Dataset & \#Nodes & \#Edges & Avg.  &Topological \\
&&Deg.&Deg.&Structure\\
 \hline
 GoogleNews& 15.8K&164.1K& 21 &News\\
 AstroPhysics &  18.7K  &216.8K&23 &Collab.\\
 Facebook SB &26.6K& 968K&73&Social\\
 Gnutella& 62.6K&210.5K&7&Peer-to-Peer\\
 Call Graph I  &631.6K&822.9K&3&Call Graph\\
BerkStanford&685.2K&7.3M&21&Web\\
 Facebook NY &905.7K&10.6M&23&Social\\
 DBLP&   1.1M&4.7M&9&Collab.\\
  Call Graph II  &47.2M&329.3M&14&Call Graph\\
 \hline
 FlatGrid &10.0K&19.8K&4&Grid\\
 HyperGrid &29.3K  &51.6K &4&Hyperbolic\\
\hline
\end{tabular}
\label{table:graphstats1}
\vspace{-0.1in}
\end{footnotesize}
\end{center}
\vspace{-0.2in}
\end{table}

\section{Related Work}
\label{sec:rel_work}

\para{Theoretical Bounds on General Graphs.}
There is a rich body of literature on distance oracles for general graphs, including the special 
case of distance labeling scheme where the distance for query node pairs is 
estimated by merely using labels associated with the query nodes, and the
related problems of graph spanners. The seminal work of Thorup and
Zwick~\cite{Thorup05} described a distance oracle that gives $2k-1$
approximation with $O(k)$ query time, $O(k n^{1+1/k})$ space and $O(k
m n^{1/k})$ preprocessing time on an arbitrary weighted undirected graph with 
$n$ nodes and $m$ edges, for any integer $k \geq 2$.
The preprocessing time and the query time of this distance oracle were
subsequently improved (c.f.~\cite{Baswana06,Kavitha06,Baswana08,Mendel06}), 
but the 
space versus approximation factor trade-off has remained almost the same.
In fact, various lower bounds have been proved under plausible
conjectures (c.f.~\cite{Patrascu10} and the references therein) for
the space versus worst-case approximation factor trade-off. These
lower bound results suggest that it is unlikely that a distance oracle
can result in a significantly better trade-off for general weighted graphs.

\para{Empirical Work on Road Networks.}
In contrast to the theoretical work on general graphs, there has been
considerable algorithm engineering and experimentation work on
road networks for navigation applications using global positioning
systems (GPS). These solutions (e.g.,~\cite{Bast09,Delling11,Delling11b,Sanders12}) crucially rely
on many specific characteristic properties of road networks such as
the existence of small natural cuts, a grid-like structure, highway
hierarchies, guiding the search towards the target using the latitude
and longitude of the target location, etc. 
On other graph classes such as those  from
online social networks, equivalent solutions are not generally known
to produce equally good approximations.

\para{Approximate Distances on Real-World Networks.} 

\noindent Distance oracles have been investigated 
both from theoretical and practical perspectives.
As described earlier, our focus is on 
distance oracles that provide accurate distance estimates
on large real-world graphs rapidly (a few microseconds) 
while having scalable 
preprocessing and near-linear storage requirement. This 
category includes a number of recent practical
heuristics \cite{orion,rigel,
20BFS,Tretyakov11,Qiao11,qi13}) that aim to estimate 
distance by embedding the graph in geometric spaces, like 
Euclidean or Hyperbolic, or by extracting different kinds 
of approximating trees from the real graph.
Other heuristics, such as, \cite{Tretyakov11,20BFS} or the 
landmark based approaches with diverse
seeding strategies~\cite{potamias09} 
use variants of 
breadth first search (BFS) trees. 
We also remark that the sketch-based
distance oracle by Sarna et al.~\cite{Sarna10}, that 
engineers a distance oracle with provable multiplicative guarantees, is similar in nature to other approaches 
cited above. 
%Even though these 
%techniques look quite promising in theory, in practice 
%they do not provide competitive 
%bounds on the quality of the distance estimations. 
On the other side, there are some theoretical
approaches~\cite{chen12,microsoft,gromov,Frenchies} that prove 
worst-case bounds on accuracy for specific  
graph classes such as those with power-law degree distribution or
those with small graph 
hyperbolicity~\cite{CD00,gav-ly,chepoi}.
These techniques have not been evaluated on large 
real-world graphs. We evaluate some of these
techniques for benchmarking of our oracle's
performance. 

\para{Exact Distance Oracles.} 
Since exact distance oracles require the
computation or storage of all pair shortest path, 
we do not consider these
(e.g.,~\cite{akiba13,abraham12} and the references therein) 
in our
comparison. We observe that even with the best combination 
of engineering
insights, all-pairs shortest path computation remains far 
too slow for
large graphs with $\sim 10^7$ nodes or edges and beyond, 
and this is especially 
unmanageable when these graphs do not fit in the main 
memory of the computing device. 
Also, we do not consider the
oracles that have 
$\Omega(n^{\epsilon}), \epsilon>1$ query complexity
(e.g.,~\cite{agarwal13,agarwal13b,porat11}), as these are 
unlikely to
yield efficient solutions.

\section{A Geometric Distance Oracle}
\label{sec:hyperbolicity}
In this section, we present an overview of 
graph hyperbolicity and discuss how graph hyperbolicity 
may be exploited to 
design an effective distance oracle.
Recent studies~\cite{ns11,inria13,kennedy13} show that 
large-scale networks, from IP-layer connectivity, citation,
collaboration, co-authorship and friendship graphs, exhibit
strong intrinsic hyperbolicity.  
\begin{figure*}[t]
\centering
\subfigure[3-point condition]
{
\includegraphics[trim=0cm 13cm 15cm 0cm, clip=true,scale=0.45]{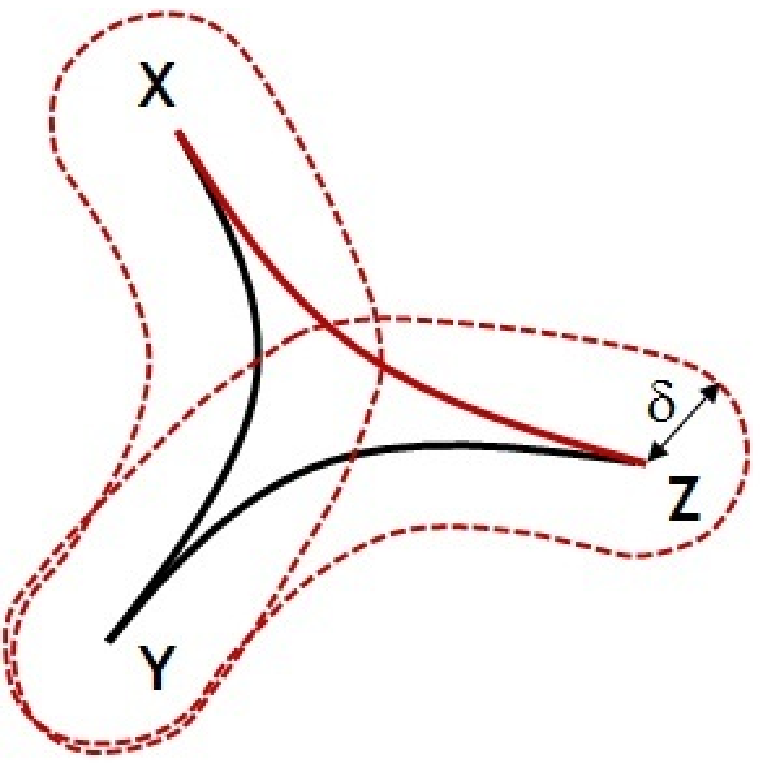}
\label{fig:3-point}
}
\subfigure[4-point condition]
{
\includegraphics[scale=0.6]{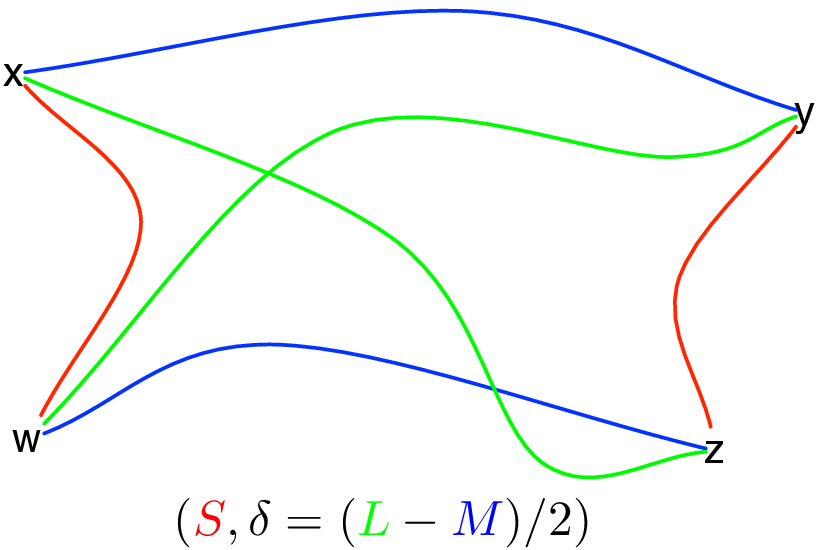}
\label{fig:4-point}
}
\subfigure[2$\delta$ distortion relative to trees]
{
\includegraphics[scale=0.75]{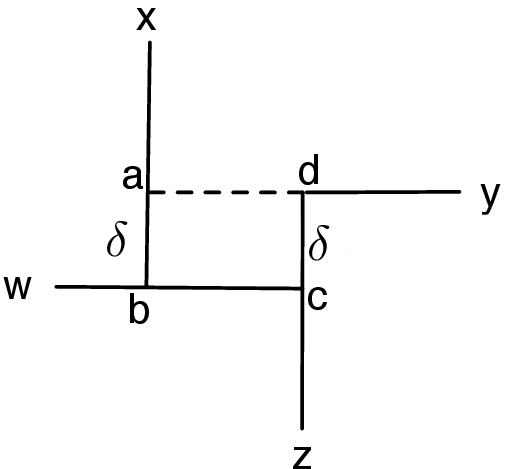}
\label{fig:2d-distortion}
}
\caption{(a) A schematic representation of the 
simple-to-visualize 3-point condition.  
(b) Simpler to compute is the 
the equivalent 4-point condition, 
where for any four points $w,x,y,z$, 
the largest sum $L$ of opposite pairs 
($xz$+$yw$ in this picture) minus 
the middle sum $M$ (=$xy+wz$ here) is no more than 
twice $\delta$. (c) The 4-point condition 
implies a $2\delta$ deviation from
a tree graph, in which the largest among the three 
sums of opposite pairs of distances of four points
is always equal to the medium sum.}
\label{fig:delta-hyperbolic-computation}
\vspace{-0.1in}
\end{figure*}

\para{$\delta$-Hyperbolicity}
Intrinsic hyperbolicity, by which is meant hyperbolicity
without any embedding of the graph into some Euclidean and 
other space, 
captures the geometric notion of negative curvature in smooth geometry. In the past two decades, the notion 
of negative curvature has 
been successfully exported to metric spaces, 
which are more general and less restrictive than
Euclidean spaces.
% in the form of $\delta$-hyperbolicity.  
This generalized notion of curvature lends itself 
naturally to the investigation of
(large-scale) curvature in graphs \cite{ns11, kennedy13}.
A simple description of $\delta$-hyperbolicity in a
(path) metric space is
that the three sides of \emph{any} shortest-path
triangle $X,Y,Z$ always come within
a certain fixed distance $\delta$ of each other,
where $\delta$ is a fixed minimal constant 
associated with 
the graph.  
In other words, 
the union of two $\delta$-neighborhoods of any two
sides of a shortest-path triangle includes the third, 
as depicted in Figure~\ref{fig:3-point}.

An easier to compute equivalent condition is
as follows. Given a graph $G(V,E)$ and any four nodes $w,x,y,z \in V$, consider
the three sums of distances between opposite pairs of nodes.  
Specifically, let $S \defeq S(w,x,y,z)=d(w,x) + d(y,z) \leq M \defeq M(w,x,y,z)=d(x,y) + d(w,z) \leq
L \defeq L(w,x,y,z) =d(x,z) + d(w,y)$ as shown in 
Figure~\ref{fig:4-point}.
Here, $d(x,y)$ is the shortest path distance 
between nodes $x$ and $y$ in $G(V,E)$; 
often we will assume the standard `hop' metric on $G$.
The hyperbolicity $\delta$ of a graph may be defined~
\cite{gromov} as follows:
\begin{definition}
A graph $G(V,E)$ is said to be
$\delta$-hyperbolic for some $\delta \geq 0$ if for every 
four nodes $w,x,y,z \in V$, 
$(L - M)/2 \leq \delta$ is always
satisfied (i.e., the largest two of the three sums of 
opposite side distances differ by no more than $2 \delta$).
\end{definition}
We refer to $\delta$ as the (4-point) hyperbolic 
constant of the graph\footnote{Even though $\delta$ is used
to refer to both the 3-point and 4-point constants, 
for the same
graph these two need not have the same value. From here on 
we use $\delta$ to denote the 4-point constant 
of the graph.}.
Notice that any finite graph with diameter $\Delta$ is 
trivially $\Delta$-hyperbolic.
Thus this definition is insightful only when $\delta$ is 
considerably smaller than $\Delta$, as argued, for example, 
in \cite{ns11,kennedy13}. In Table~\ref{table:graphstats2} 
we list the
estimated $\delta$ values on our benchmark graphs to show 
that indeed
on real world networks the $\delta$ values are actually 
very small. 
For a 
detailed description of how $\delta$ may be estimated
for large graphs, see \cite{ns11,kennedy13,inria13}. 
\begin{table}[t]
\begin{center}
\caption{\small Key Geometric Characteristics of the $9$ Real 
and $2$ Synthetic Benchmark Graphs (based on large samples).}
\vspace{0.1in}
\begin{footnotesize}
\begin{tabular}{|c|c|c|c|c|}
\hline
%  && \multicolumn{2}{|c|}{Scalar  values} \\ \cline{3-4}
   %& Statistical &\multicolumn{7}{|c|}{Distribution Models} \\ \cline{3-9}
   Dataset & Clust Coeff. & $\max \delta$ &avg $\delta$& Diam\\
 \hline
% Facebook SB &12,814&92,241& 13&0.24&&Social \\
 GoogleNews& 0.45&1.0&0.03&4\\
 AstroPhysics &0.50&2.0&0.24&\textasciitilde5\\
 Facebook SB &0.22&2.0&0.23&14\\
 Gnutella&0.01&1.5&0.08&\textasciitilde5\\
 Call Graph I &0.14&4.5&0.39&\textasciitilde47\\
BerkStanford&0.50&2.0&0.16&\textasciitilde5\\
 Facebook NY &0.16&2.0&0.22&\textasciitilde19\\
 DBLP&0.65&2.5&0.24&\textasciitilde23\\
Call Graph II &0.09&$<4$&0.27&\textasciitilde27\\
 \hline
 FlatGrid &0.00&87.0&7.40&198\\
 HyperGrid&0.33&1.0&0.27&18\\\hline
\end{tabular}
\label{table:graphstats2}
\end{footnotesize}
\end{center}
\vspace{-0.2in}
\end{table}

\para{$\delta$-Approximation of Hyperbolic Graphs with Trees.}
It has been observed that the distance metric of a 
$\delta$-hyperbolic graph can be 
viewed as a $\delta$-approximation of a tree metric, since for any four points 
$w,x,y,z \in V$ in a
tree $T(V,E')$, (for an appropriate set $E'$ of edges)
one necessarily has that $S \leq M =
L$, thus resulting in $\delta=0$.\footnote{Note that another notion of tree-likeness, 
tree-width, is an independent property and 
unlike $\delta$-hyperbolicity its existence
in real-world graphs is not common. For instance, 
snapshots of
the internet at autonomous system level have large tree-width, but low hyperbolicity~\cite{montgolfier11}.}

While trees and block graphs are
$0$-hyperbolic and chordal graphs have $\delta \leq 1$, cycles $C_n$ are $O(n)$-hyperbolic and square grids 
with $n$ nodes are $O(\sqrt{n})$-hyperbolic.

It turns out that similar to trees, $\delta$-hyperbolic graphs have a non-empty core of nodes
whose betweenness centrality is maximal possible, that is, they have nodes whose betweenness
centrality scales as $\Theta(n^2)$ (where $n:=|V|$ is the size of the graph), 
as argued in~\cite{ns11,jonk11}. 
In other words, $\delta$-hyperbolicity also 
captures the
notion that there is always a non-empty set
of vertices, its `core',
where a fraction of \emph{all} shortest paths pass through.\footnote{Note that our usage of the term `core' here
is different from the core of online social networks (that contains most of the
graph nodes) and the notion of coresets used in streaming algorithms literature.}
This is a crucial property that we shall leverage in
our distance approximation oracle.

\para{Theoretical Bounds via Tree Approximations.}
We focus on three classes of distance 
oracles that approximate
distances on graphs via trees. First, 
Gromov-like techniques, 
principally~\cite{gromov,gav-ly,chepoi}, 
yield an additive guarantee on the 
distortion on distances via tree 
approximations of the graph, where the resulting trees 
are not sub-graphs of the original graph but are
typically (based on) contractions. 
Specifically, Gromov's notion of
hyperbolicity leads to approximations of a graph $G$
with trees $T$ that satisfy 
the following bound: For any two points
$x,y \in V(G)$, the shortest path distance 
$d_{GT}(x,y)$ between them in $T$ has distortion 
$|d_G(x,y) - d_{GT}(x,y)| \le 2\delta \log{n}$. 
Second, the tree embedding proposed 
by Abraham {\em et al}~\cite{abraham12}.
yields a multiplicative guarantee on the distortion for 
distances in the tree. A metric space $X$ is defined to be $\varepsilon$-hyperbolic, 
$\varepsilon \in [0,1]$, 
if every set of four points $w,x,y,z$ in $X$ ordered 
so that $d(w,x)+d
(y,z) \le d(w,y)+d(x,z) \le d(w,z)+d(x,y)$ satisfies
$$d(w,z)+d(x,y) - d(w,y) - d(x,z) \le 2\varepsilon\min \{d(w,x),d(y,z)\}.$$
For such a metric space the authors give an algorithm to construct an 
embedding $X$ into a tree $T$ where for any two points $x,y \in X$ the shortest path distance $d_T(x,y)$ between them in $T$ satisfies $\frac{\max_{x,y} (d_T(x,y)/d_G(x,y))}{\min_{x,y} (d_T(x,y)/d_G(x,y))} \le (1+\varepsilon)^{c_1 \log |X|}$.
Further, up to the constant $c_1$, they show that there 
exists an
example where this inequality is tight. Finally, 
Chepoi and Dragan~\cite{CD00} construct contraction
trees (i.e., many nodes of $G$ may map to 
the same node in $T$) that approximate the graph 
distances with an additive distortion not exceeding 
$|c(G)/2| + 2$ where $c(G)$ is the length of the longest 
chordless cycle in $G$.
This kind of approximation is also appropriate 
to consider when real-world graphs only have 
relatively short chordless cycles.

\section{Algorithmic Design }
\label{sec:oracle} In this section, we present an 
algorithmic framework for designing distance oracles 
that leverage the small
hyperbolicity of real-world graphs.  We exploit the hyperbolic core~\cite{ns11,jonk11}
to construct a spanning tree BFS sub-graph approximation for
the graph, and not tree contractions 
with non-graph edges as done in the 
above-referenced three approaches.
Since, as cited, for $\delta$-hyperbolic graphs, asymptotically 
a fixed fraction of all
shortest paths in the graph pass through its 
non-empty hyperbolic core, picking a node from the core as
the tree root will give us zero distortion for all
node-pair distances whose shortest paths traverse 
that node.
The node with the highest betweenness centrality is
thus ideal as the root of the spanning tree:
The smaller the hyperbolic constant $\delta$
of the graph, the tighter the core, and the
larger the fraction of all node-pair paths that
traverse a typical node in the core,
until for $\delta=0$ we have a single node to pick
as the root to ensure a fixed fraction of
all node-pair paths traverse it. The proposed tree 
oracle therefore
starts from nodes with the highest 
(betweenness) centrality.  The tree then expands
as a BFS tree by adding unvisited nodes prioritized by their 
(betweenness) centrality 
values until all nodes are included in the tree. 

Now, since computation of
centrality is too expensive on large-scale graphs,
we need a good proxy for centrality that is
easy to compute.  Empirically, it has been
observed on real-world graphs that
nodal degrees correlate (significantly) well 
with (betweenness) centrality~\cite{ns11}, and 
Figure~\ref{fig:deg_corr} 
shows a typical correlation chart for the FacebookSB
network.
A degree-based prioritization is then employed both 
during the selection of initial root of the tree as
well as during its expansion. 

\begin{figure}
\centering
\includegraphics[scale=0.3,angle=270]{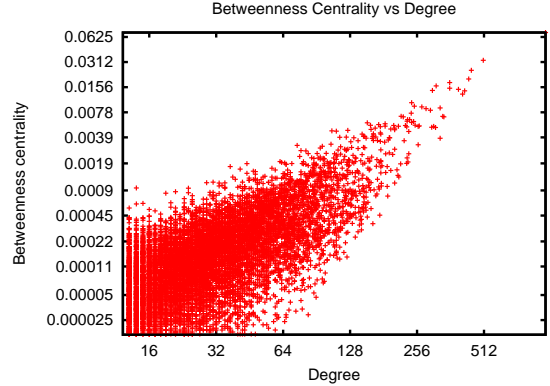}
\caption{Degree and betweenness centrality of nodes in the Facebook SantaBarbara graph.}
\label{fig:deg_corr}
\vspace{-0.1in}
\end{figure}

\para{Four Benchmarked Oracles.}
In the remainder of this section, we describe three 
distance oracles that either have
proven bounds on distance distortion
or are best-of-class empirically.  We compare 
the results with our proposed geometric oracle
on the eleven samples of real networks 
and the two synethetic graphs, as shown in
Table~\ref{table:graphstats2}.
More specifically,
we examine 1) a representative approach 
from Gromov-type contraction-based tree approximation 
with proven bounds: \textbf{Gromov Tree}, see
Figure~\ref{fig:ex0}(c),
2) an oracle based on
Steiner trees with proven multiplicative bound:
\textbf{$\varepsilon$-approximate Steiner Tree}, 
3) an empirical best-of-class landmark-based 
approach that exploits hyperbolic embedding:
\textbf{Rigel}, 
and 4) our proposed centrality-based spanning tree 
oracle: \textbf{HyperBFS}, see
Figure~\ref{fig:ex0}(e),
as outlined above. 

\subsection{Hyperbolicity-based Tree Oracle}
\label{BFSapproach}
A $\delta$-hyperbolic graph may be viewed as an 
approximation to a tree.  We thus
aim to extract a `backbone spanning tree' 
(or a small collection of such
trees for better distance approximation) 
from the input graph to construct our geometric
oracle. As stated earlier, we expect a BFS spanning tree
with a highly central vertex as root would closely
approximate distances in a $\delta$-hyperbolic real-
world graph (cf. Section~\ref{sec:results} for more
details). Also as argued earlier, we may use degree as a 
proxy for centrality
for computational efficiency and select the highest 
degree node as 
the root.
\begin{figure}
\centering
\includegraphics[scale=0.35]{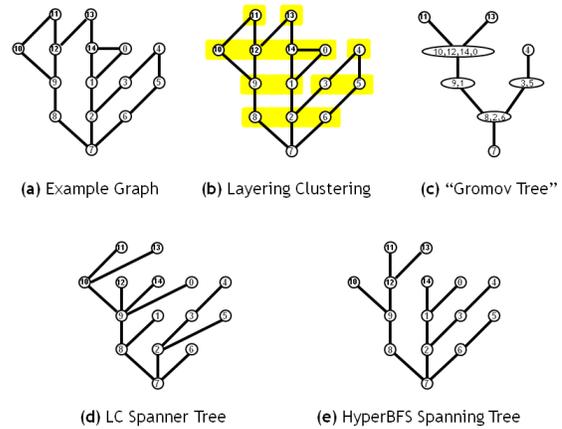}
\caption{(a) An example graph, (b) its layering partitioning~\cite{CD00}, (c) the resulting ``Gromov Tree'' approximation, (d) the associated 
tree spanner~\cite{chepoi}, 
and (e) our Hyper BFS spanning tree.  
As it may be seen from this simple example, 
Hyper BFS is distinct from (c) and (d), while 
(c) and (d) are very closely associated.}
\label{fig:ex0}
%\vspace{-0.2in}
\end{figure} 

%A breadth-first search tree is constructed around this root node.
Therefore, in constructing the Hyper BFS spanning tree 
$T$, we choose the order in which new nodes are added 
to the BFS tree strictly based on their degree.
Algorithm~\ref{alg:BFS} gives pseudocode for constructing the Hyper BFS tree based on a general vertex ordering 
$\pi$, which in our implementation, and unless stated otherwise, is based on degree.
 \begin{algorithm}
 {\bf Input:} Graph $G$, vertex ordering $\pi$\\
 {\bf Output:} Tree $T$\vspace{.1cm}\\
 {\bf let} $r$ be the first vertex in $\pi$\\
 {\bf set} $Q = \{r\}$ and $T = r$.\\
 {\bf while} $Q$ is not empty:\\
 ... {\bf set} $v$ top of $Q$ (remove $v$ from $Q$)\\
 ... {\bf let} $N = N(v)$ ordered by $\pi$\\
 ... {\bf for each} $n \in N$ not already in $T$:\\
 ...... {\bf push} $n$ on $Q$\\
 ...... add vertex $n$ and edge $nv$ to $T$\\
 {\bf return} T
 \caption{ -- Hyper BFS with Vertex Ordering\label{alg:BFS}}
 \end{algorithm}

To improve the accuracy of distance estimates, and as
practiced previously, 
we use a small number of different trees, each with a
distinct node as root. We construct a collection of 
such trees 
rooted at distinct nodes in the hyperbolic core of
the graph, which as stated, we approximate by nodal 
degree.
Finally, the distance between two nodes $x$ and $y$ is answered by returning the minimum of
the distances between $x$ and $y$ in the (small set of) different trees we construct.

\subsection{Contraction Trees with Bounded Approximation}
\label{sec:treeapprox}
We start with the `Gromov tree' technique 
which yields an approximation with guaranteed additive 
distance distortion based on a contraction tree.
Next we outline the approach due to 
Abraham {\em et al.} which 
yields a multiplicative guarantee on the distance 
distortion in an approximating tree.

\para{Gromov tree approximation.} 
Gromov's notion of hyperbolicity~\cite{gromov} naturally 
leads to a design for a linear time algorithm for 
approximating a graph $G$ with a contraction tree $T$.
For any two points $x,y \in V(G)$, the 
unique
shortest path distance $d_{GT}(x,y)$ between them in the Gromov tree $T$
satisfies $|d_G(x,y) - d_{GT}(x,y)| \le 2\delta \log{n}$.
%$d_G(x,y) - 2\delta \log{n}  \le d_T(x,y) \le d_G(x,y)$.
Following~\cite{CD00}, for a connected graph $G$, we 
fix a root $r \in V(G)$ and for $i \ge 0$,  
let $N_i(r)$ be the vertices at distance {\bf exactly} $i$ from $r$ in $G$, and 
$B_i(r)$ be the vertices at distance {\bf at most} $i$ from $r$ in $G$.
A Gromov tree $T$ is then obtained by contracting all node pairs $u, v \in N_k(r)$ (for all $k \ge 1$), such that there exists a path between $u$ and $v$ entirely contained outside $B_{k-1}(r)$. 
For convenience, we define a {\em $i$-level connected 
component} as a maximal subset $C$ of vertices in $N_i(r)$ 
such that each pair of vertices in $C$ is connected in 
$G \setminus B_{(i-1)}(r)$.   
This construction (`layer partition' of a graph), 
proposed in~\cite{CD00,chepoi} 
and also used by~\cite{gav-ly}, determines a tree $T$ 
satisfying 
the above properties in linear time by dealing first with 
those nodes furthest from $r$; if we think of 
trees with the root at the top, we can say `in a bottom-up 
manner'. 
In each level, $N_i(r)$, we contract the $i$-level 
connected components into a single node. 
The bottom-up approach yields that in order to find these 
components we need only consider the edges completely 
contained in $N_i(r)$ and those between 
$N_i(r)$ and $N_{i+1}(r)$.  
It follows that the total time checking $i$-level 
connectivity throughout the execution is linear, and 
therefore, the algorithm also runs in linear time.   

Although this Gromov-type tree
construction (per~\cite{CD00,chepoi,gav-ly})
allows arbitrary nodes to be picked 
as root, in our implementation we pick a
high centrality (with degree as proxy) node. 
We found 
that this modification considerably improves the accuracy 
of distance estimates via `Gromov type' 
contraction trees. Hereafter, we refer to 
our modified and improved implementation of prior work
(such as \cite{CD00,chepoi,gav-ly}) 
as the {\em Gromov Tree}.
Pseudocode of our implementation 
is given in Algorithm \ref{alg:Gromov}.

 \begin{algorithm}
 {\bf Input:} Graph $G$\\
 {\bf Output:} Tree $T$\vspace{.1cm}\\
 {\bf let} $r$ be high degree node\\
 {\bf determine} sets $N_0(r),\ldots, N_\ell(r)$\\
 {\bf set} $T = G$\\
 {\bf for} $i = \ell$ downto $0$:\\
 ...  {\bf for each} $i$-level connected component $C$ of $N_i(r)$:\\
 ...... contract $C$ in $T$.\\
 {\bf return} T
 \caption{ -- Gromov Tree\label{alg:Gromov}}
 \end{algorithm}

\fixme{In the construction of `Gromov tree', the distance
$d_{GT}(x,y) \leq d_G(x,y)$. This allows us to use 
the above
technique for a lower bound on the distance estimate and, 
thereby,
compute an approximation range around the real distance 
(more on this in Section~\ref{sec:approx_range}).}

\begin{lemma}
\fixme{For any two nodes $x,y$ in $G$, $d_{GT}(x,y) \leq d_{G}(x,y)$}
\label{lemma:gromov_contraction}
\end{lemma}
\begin{proof}
\fixme{A key observation here is that in the construction, 
we only
contract nodes and do not delete any edge. This contraction of nodes
may coalesce multiple edges into a single edge. Thus, for each edge
$\{u,v\} \in E(G)$, we either have that $u$ and $v$ are contracted to
the same node in $T$ or $\{C(u),C(v)\} \in E(T)$, where $C(u)$ is the
contracted nodes of $u$ in $T$. In particular, this means that for
$\{u,v\} \in E(G)$, we have $d_{GT}(u,v) \defeq d_{GT}(C(u),C(v)) \leq
1$.}

\fixme{Consider the shortest path $P(x,y)$ in $G$ (with 
number of
  edges $|P(x,y)| = d_G(x,y)$). For each edge $\{u,v\}$ in
$P(x,y)$, it holds that $d_{GT}(u,v) \leq 1$. Thus, $d_{GT}(x,y) \leq
\sum_{\{u,v\} \in P(x,y)} d_{GT}(u,v) \leq \sum_{\{u,v\} \in P(x,y)} 1
= |P(x,y)| = d_G(x,y)$.}
\end{proof}

We note in passing that~\cite{chepoi} further refines this 
contraction tree by expanding each partition back to the 
original number of nodes by connecting these 
nodes back to
a single node in the lower layer partition, as
shown in Figure~\ref{fig:ex0}.  We shall not consider this
``layer partition spanner trees'' further as these
only add/subtract a fixed amount to the ``Gromov Tree''
approximation 
(zero distance for nodes in the same partition change
to 2).

\para{$\varepsilon$-approximate Steiner tree.} A 
slightly different version
of hyperbolicity was studied by Abraham {\em et al.} in
\cite{microsoft} and it too leads to a tree
approximation but with a guaranteed multiplicative 
distortion.
As in the Gromov construction, their algorithm fixes a 
root $r \in 
V(G)$ and consider $N_i(r)$.
The tree $T$ is then constructed on vertex set 
$V(G) \cup S$, where $S$ is a set of Steiner points, 
one for each connected component in the graph induced on
each $N_i(r)$, $i > 0$.
The edge set is defined as follows.
For a connected component $C$ in the graph induced on $N_i(r)$, where  
$s_C$ is its corresponding Steiner point, we add edges $cs_C$ of 
weight $\frac{1}{2}$
for each $c$ in $C$ and edge $s_Cx$ of weight $\frac{1}{2}$, where $x$ 
is some vertex of $N_{i-1}(r)$ connected to some vertex of $C$.
Pseudocode is found in Algorithm \ref{alg:microsoft}.

 \begin{algorithm}
 {\bf Input:} Graph $G$\\
 {\bf Output:} Tree $T$\vspace{.1cm}\\
 {\bf let} $r$ be high degree node\\
 {\bf determine} sets $N_0(r),\ldots, N_\ell(r)$\\
 {\bf set} $T = \emptyset$\\
 {\bf for} $i = \ell$ downto $0$:\\
 ...  {\bf for each} connected component $C$ of $N_i(r)$:\\
 ......  add vertex $s_C$ to $T$\\
 ......  add edge $s_C x$ of weight $1/2$ to $T$ (for one $x \in N_{i-1}(r)$ which is adjacent to $c \in C$) \\
 ......  {\bf for each} $c \in C$:\\
 .........  add vertex $c$ and edge $c s_C$ of weight $1/2$ to $T$\\
 {\bf return} T
 \caption{ -- $\varepsilon$-approximate Steiner Tree\label{alg:microsoft}}
 \end{algorithm}

\subsection{Embedding in Geometric Space}
\label{sec:techniques}
In this section, we present oracles that embed a 
graph into some
multi-dimensional geometric space.   
We pick these oracles for 
special
attention since they typically outperform other distance
approximation oracles and thus help benchmark the 
fidelity of our oracle.
Hyperbolic embedding involves explicitly mapping the 
nodes of the graph into points in the hyperbolic space.
(We observe here once more that intrinsic hyperbolicity
is not the same as hyperbolic embedding used in these distance oracles based on~\cite{kleinberg, rigel}.)

Some oracles in this category, such as
Orion~\cite{orion} and that of Qi et al.~\cite{qi13} 
use the $L_2$ norm
distance function in $\mathbb R^{10}$. Others such as Rigel~\cite{rigel} use
a {\em Hyperboloid} model~\cite{BBS_Hyperbolic} with curvature $c$,
where the distance between two $d$-dimension points $x=(x_1,x_2,\ldots,x_d)$ and
$y=(y_1,y_2,\ldots,y_d)$ is defined as follows:
\begin{small}
\begin{equation*}
\label{eq:dis}
arccosh\left(\sqrt{(1+\sum_{i=1}^d{x_i^2})(1+\sum_{i=1}^d{y_i^2})}-\sum_{i=1}^d{x_iy_i}\right)
\cdot |c|
\end{equation*}
\end{small}

 \begin{algorithm}
 {\bf Input:} Graph $G$\\
 {\bf Output:} Coordinates for all nodes in $\mathbb
 R^{10}$\vspace{.1cm}\\
%{\bf let} $Ord = \{\}$\\
{\bf Select} a set of high degree nodes $L$\\
{\bf for each} $v \in L$:\\
... BFS$(v)$\\
Compute coordinates for $v \in L$ minimizing distortion from each other\\
{\bf for each} $v \in V \setminus L$:\\
... {\bf Select} a subset $L_v$ of $L$\\
... $Ord(v) = $ Coordinates that minimize distortion from $L_v$\\
 {\bf return} Ord
  \caption{ -- Rigel \& Hyperbolic Embedding of Graph\label{alg:rigel}}
 \end{algorithm}
 
These oracles embed a graph $G$ 
by first identifying a small set of
nodes with high degrees that are called 
landmarks. Using BFS from all
landmark nodes, they compute the distance of all 
nodes in the graph
from these landmark nodes. Then, a linear program is defined to
compute an embedding of these landmark nodes in $\mathbb R^{10}$ such
that the difference between the actual pairwise distance of these
landmarks and their distance estimated using the defined distance
function in the embedded space is minimized. Solving this 
linear program (e.g., via the simplex method) provides 
the coordinates of landmark nodes.

The remaining nodes in $G$ are then given a coordinate 
using another
linear program. The objective of this linear program is to 
minimize
the difference between the actual distance of the node 
to a subset
$L_v$ of landmark nodes and the distance estimated using 
the distance
function in the embedded space. Once again, the simplex 
method is used to
solve the linear program. Algorithm~\ref{alg:rigel} 
provides a
pseudocode for the Rigel approach.

The approximate
distance among each pair of nodes is then 
their distance in the
embedding. Thus, the query time and space per node only 
depends on the
dimensionality of the embedding and for a fixed dimension, 
it reduces to $O(1)$ (possibly with a high value for 
the prefactor). 

We remark that a plausible reason for the success of 
these hyperbolic embedding techniques is that
in real-world graphs, the set of `core nodes' of 
high
centrality (as implied by small graph hyperbolicity) and 
the set of
landmark nodes (computed based on degree) includes many of 
these same nodes as shown in Figure~\ref{fig:deg_corr}.

Note that these coordinate-based systems can both 
underestimate and overestimate the real
distances. Rigel has been shown to be 
significantly more accurate 
than high-dimensional
Euclidean embeddings~\cite{rigel} and our preliminary 
experiments also
confirmed the same. Hence, we only consider Rigel from this 
category
in our empirical comparison.

\begin{figure*}[!htb]
\subfigure[Santa Barbara, Facebook]
{
\epsfig{figure=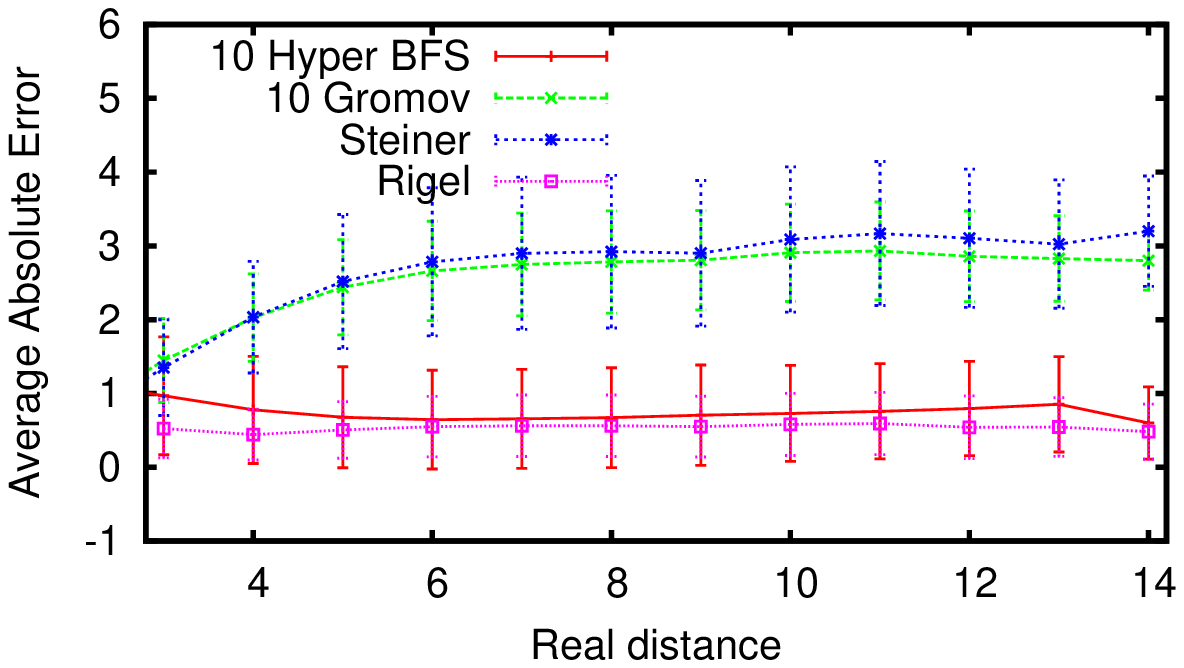,width=2.35in}
\label{fig:avAbsSB}
}
\subfigure[Call Graph]
{
\epsfig{figure=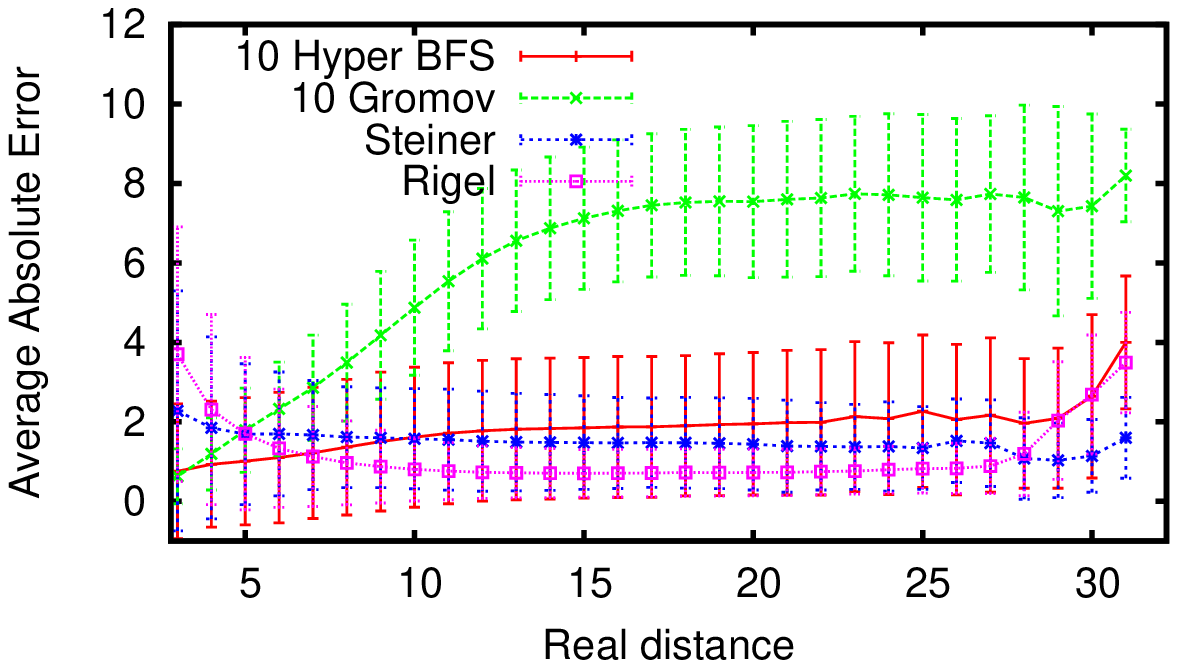,width=2.35in}
\label{fig:avAbsCall}
}
\subfigure[Hypergrid Graph]
{
\epsfig{figure=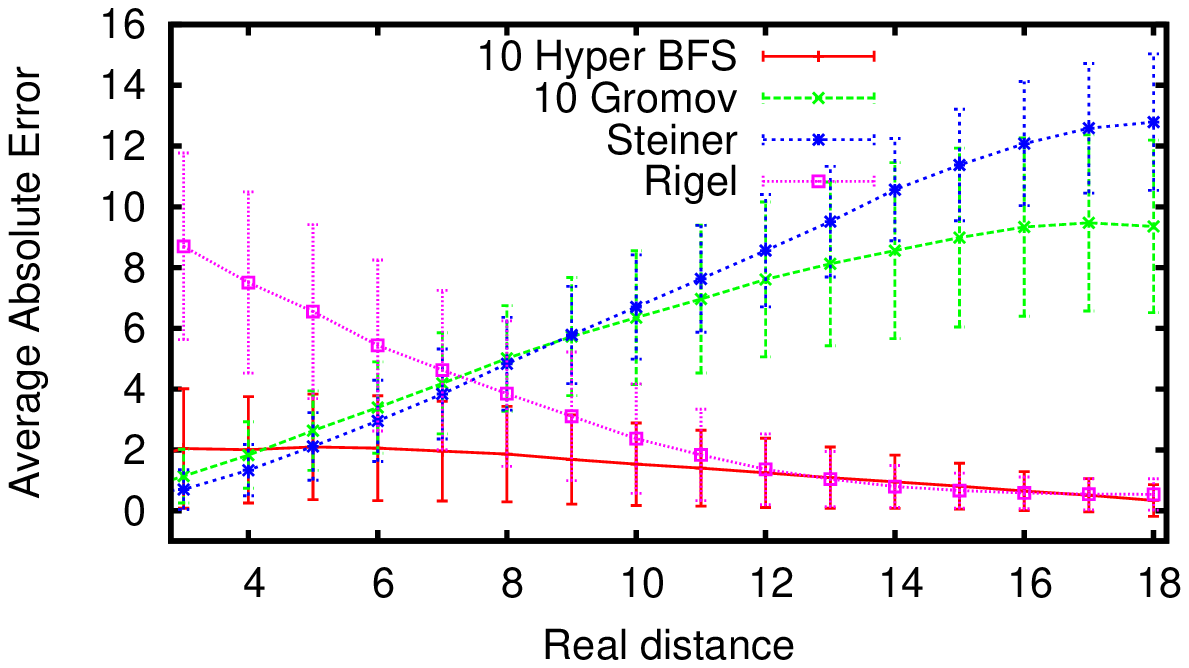,width=2.35in}
\label{fig:avAbsHyper}
}
\subfigure[P2P Gnutella]
{
\epsfig{figure=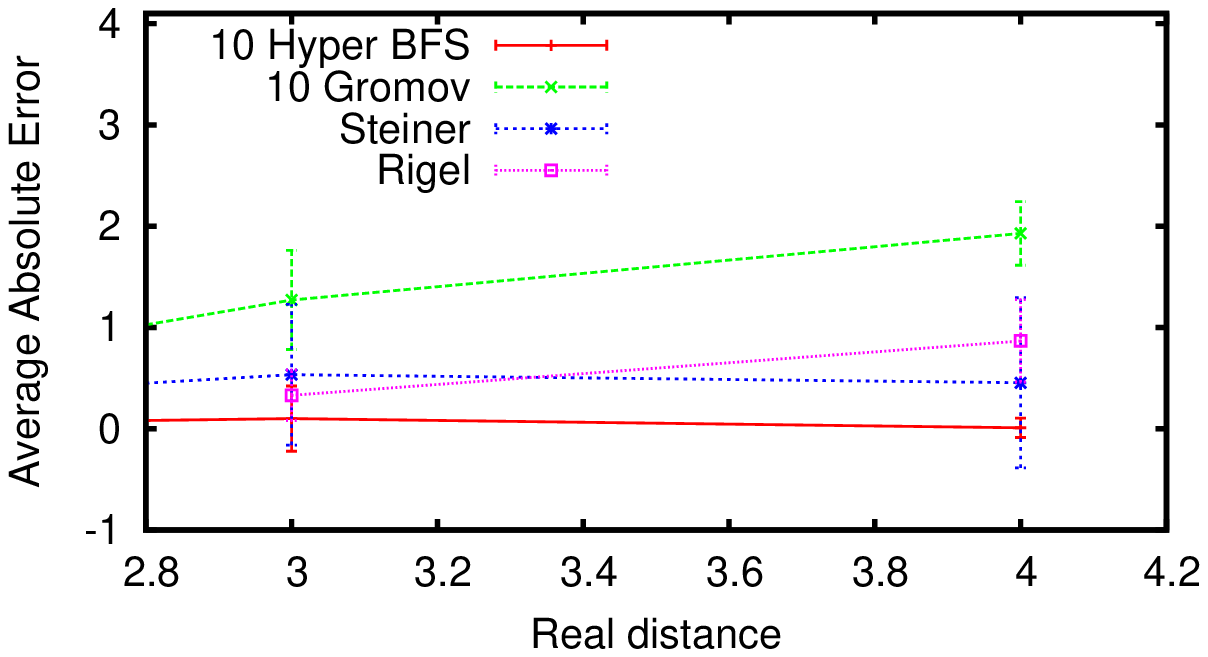,width=2.35in}
\label{fig:avAbsP2pGnutella}
}
\subfigure[Web Berk-Stan]
{
\epsfig{figure=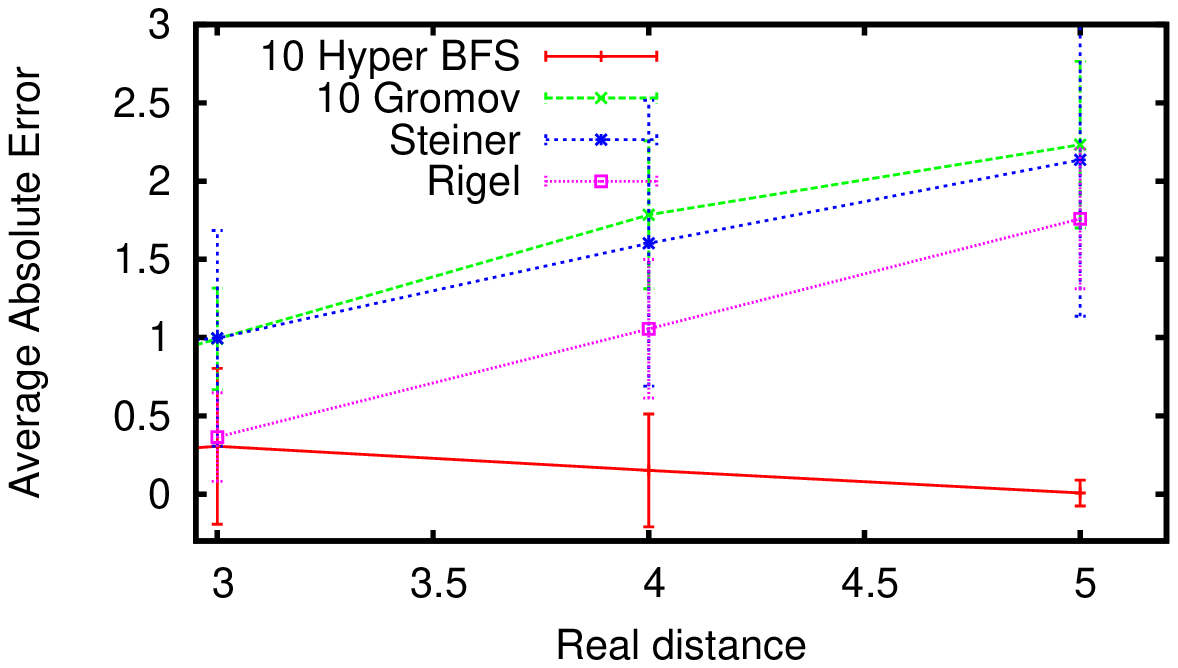,width=2.35in}
\label{fig:avAbsWebBerkStan}
}
\subfigure[Google News]
{
\epsfig{figure=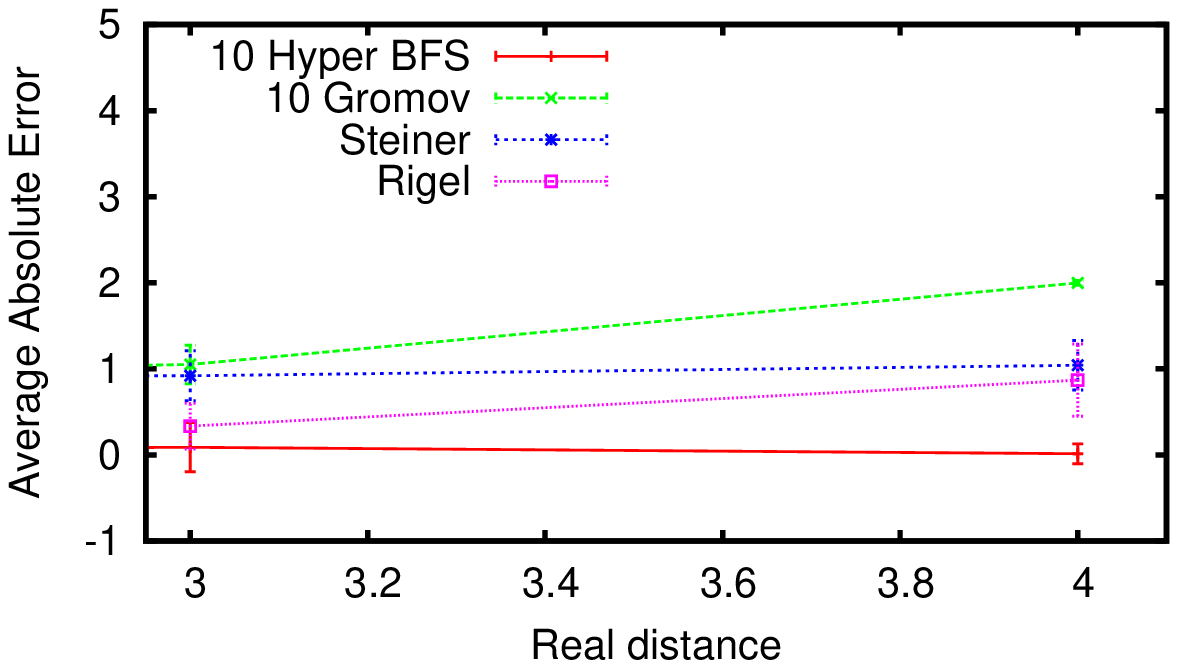,width=2.35in}
\label{fig:avAbsGoogleNw}
}
\subfigure[CaAstro-Ph]
{
\epsfig{figure=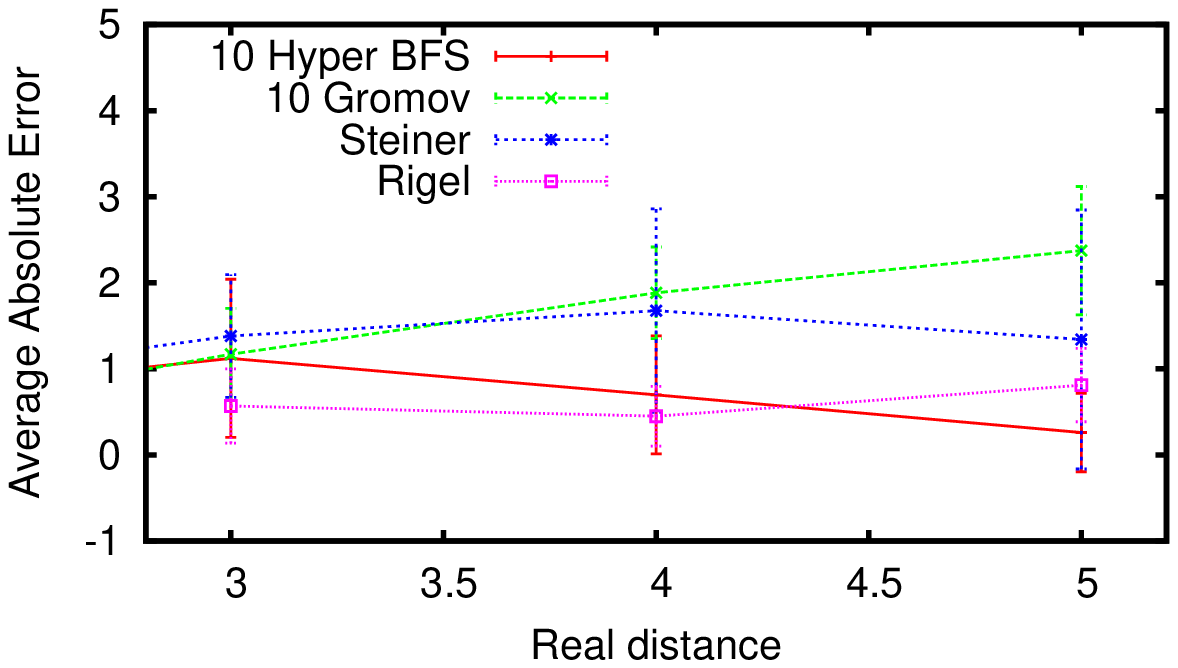,width=2.35in}
\label{fig:avAbsCaAstroPh}
}
\subfigure[DBLP]
{
\epsfig{figure=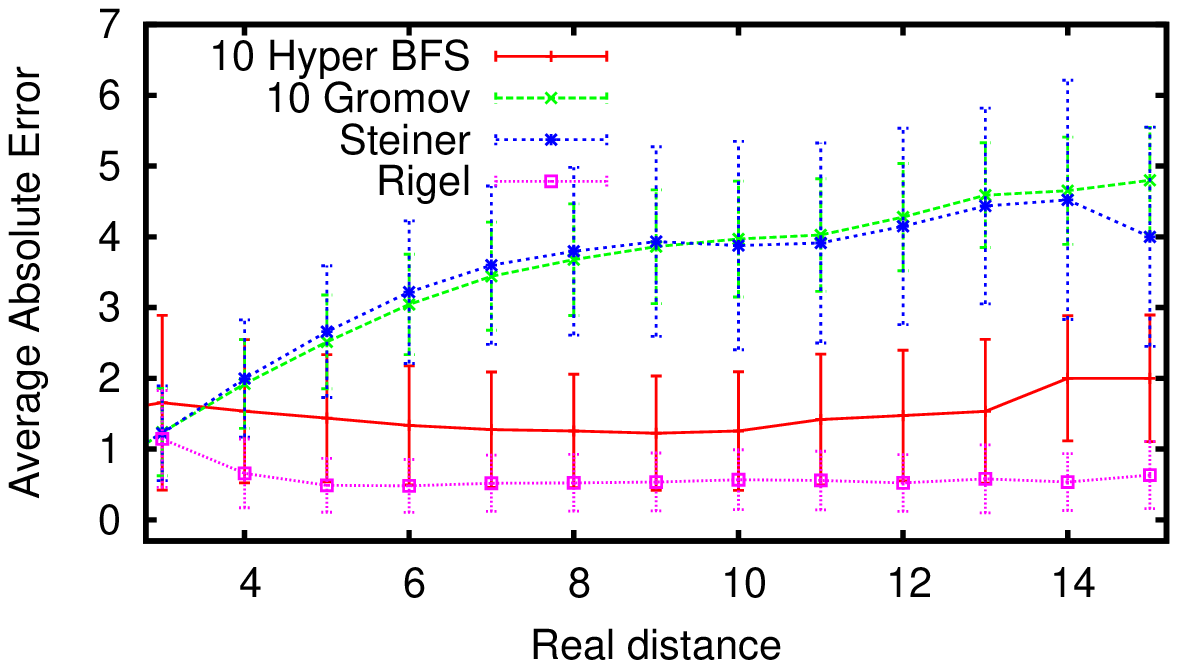,width=2.35in}
\label{fig:avAbsDBLP}
}
\subfigure[Facebook NY]
{
\epsfig{figure=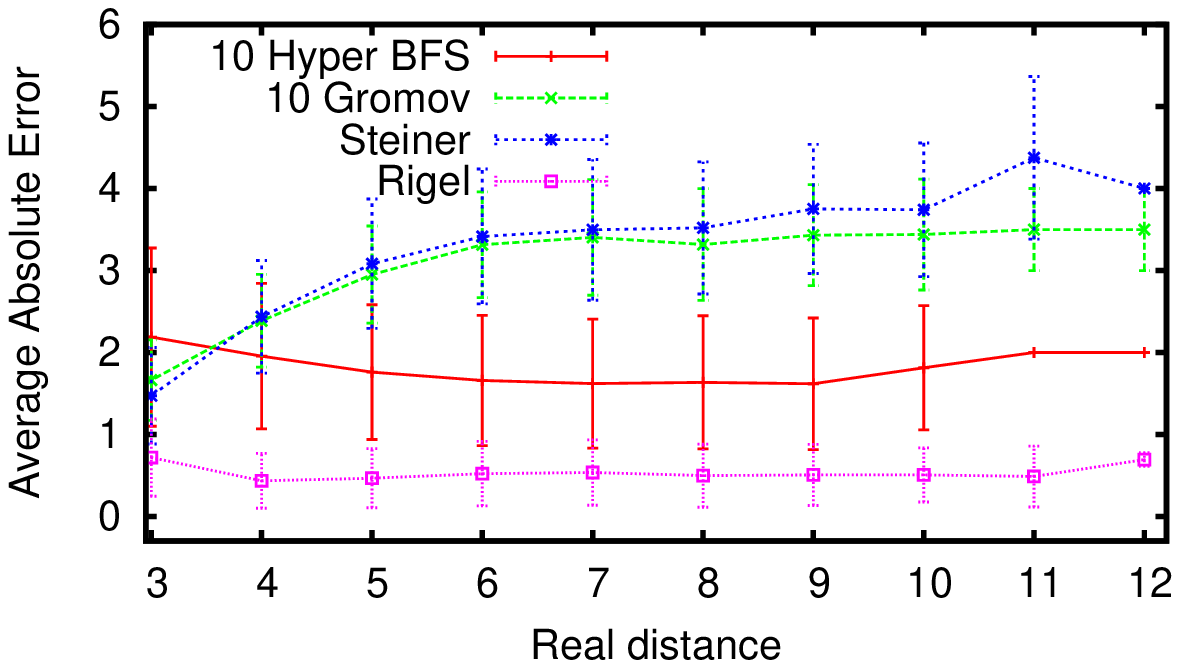,width=2.35in}
\label{fig:avAbsFacebookNY}
}
\caption{Average absolute error of various approximation techniques on
  various synthetic and real-world graphs.}
%\caption{Average absolute error of various approximation techniques on two real-world  graphs and one synthetic graph.}
\label{fig:avAbsError}
\end{figure*}

\section{Experimental Evaluation}
\label{sec:results}
In this section we describe the methodology used in our 
experiments, followed by a detailed analysis.
The goal of our benchmarking is to determine
if the observed $\delta$-hyperbolicity 
of real-world or synthetic networks
yields a low cost and effective distance approximation
competitive with or better than the 1) best-of-class or
2) theoretically guaranteed approaches, as discussed. 

\subsection{Experimental Setup}  
The networks we study range from 10s of thousands to
100s of millions of nodes and edges.  
Storing $(|V|(|V| - 1))/2$ distances between 
all node pairs results in
Gigabytes of storage for small size
graphs and quickly hits hundreds of Terabytes for our 
larger graphs, which is impractical.  
To ensure that accuracy is measured 
over a large enough 
sample of all node pairs in a graph, we use two 
different approaches. 
When the studied graph is small enough, we compute the
%%average 
exact distortion of distance for each pair of nodes. 
For large
graphs we compute the distance distortion by
sampling a large enough set of node pairs to
ensure that the mean error due to sampling is 
sufficiently small.  
 
A few observations are in order.
First, for each of the Hyper BFS and Gromov Tree 
techniques,
we construct not just one tree but a small
collection of trees rooted at nodes 
in the hyperbolic core of the graph, i.e., 
those with the highest centrality, approximated
with degree centrality for computational efficiency.  
Once we obtain this collection of trees, the distance
between two nodes $x$ and $y$ is answered by returning 
the minimum of
the distances between $x$ and $y$ in the different 
trees of the Hyper BFS. 
For the Gromov Tree, the maximum distance gives a
better lower bound to the graph distance, since  
Gromov Trees are contractions.  
Second, we have verified experimentally that
$10$ such trees are enough to provide very good distance
approximations for the Hyper BFS and our
implementation of Gromov Trees. 
This is an improvement relative to
prior work~\cite{20BFS}, where it was observed
experimentally that $20$ of their spanning trees were
needed for a similar level of accuracy.  This 
improvement is likely due to our selection of roots 
from the centrality core of the graph.
Finally, although we have computed and carefully 
examined all four 
oracles on all 11 benchmarked graphs, in what 
follows we present only the most representative 
plots. 
%%%%%%%%%%%%%%%%%%%%%%%%HHHHHHHHHHHHHHHHHHH

\para{Measures of Distortion.}
We use three measures for evaluating the performance of the various 
distances approximations on each graph.
Each of these measures compares the {\em ground truth} graph 
distance $d_G$ with the approximate distance given by the proxy $d_A$ 
and captures a different performance feature.  
\begin{definition}
Let $x,y$ be vertices of a graph $G$ and let $d_A$ be the distance
approximated by a distance oracle. 
\begin{itemize}
\item The {\em additive distortion} between $x$ and $y$ with respect to $d_A$ is $d_G - d_A$.
\item The {\em absolute distortion} between $x$ and $y$ with respect to $d_A$ is $|d_G - d_A|$.
\item The {\em multiplicative distortion} between $x$ and $y$ with respect to $d_A$ is $\frac{|d_G-d_A|}{d_G}$.
\end{itemize}
\end{definition}
The selection of proper distortion metrics is critical in assessing
the quality of each distance approximation model since each metric
captures distinct aspects of the distortion error and may have
different consequences for different applications. The choice of a measure is also important in the context of different kind of bounds
known for different approaches. For instance, the Gromov tree approach
has additive guarantees while $\varepsilon$-approximate Steiner tree has
multiplicative guarantees on distance errors.
We visualize each of these distortion metrics by plotting its expected
value conditioned on the value of $d_G$ being fixed. This allows us to
study the accuracy of various techniques for node pairs at different
distances. This is important as applications, like graph segmentation, SVD
decomposition, recommendation systems or influential node detection,
may require different accuracy guarantees across the range of node
distances, i.e. neighbor nodes, diametrically opposite nodes or node
pairs with distances in between the two extremes. 
For instance, recommendation systems may require highly accurate 
short range distances around a particular user to detect similar 
users in a particular radius; on the contrary, 
influential node detection would require accurate long distances 
to determine the centrality of a particular node.
 \begin{figure}
\centering
\includegraphics[scale=0.6]{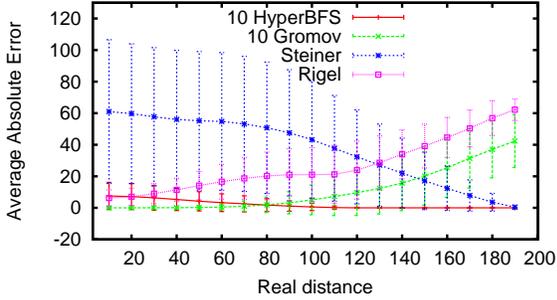}
\caption{Average absolute error of various approximation techniques
on a SquareGrid network}
\label{fig:avAbsFlat}
\vspace{-0.1in}
\end{figure}
\vspace*{0.4cm}
\subsection{Methods' Fidelity}
\label{sec:results_fidelity}

\para{Comparing Oracles Based on Absolute Error.}
Our first key observation is that the error 
due to our oracle in estimating shortest path
distances is very small. The 10-Hyper BFS approximation
results in an absolute error of less
than 2 in all instances.
Figure~\ref{fig:avAbsError} shows 
results on eight real graphs and one synthetic graph, HyperGrid. 
We next compare the accuracy of the Hyper BFS approach with the best practical 
oracle system, i.e. Rigel. This system is particularly optimized for networks 
with Power-Law degree distribution like Santa Barbara(Figure~\ref{fig:avAbsSB}) 
and a Call Graph (Figure~\ref{fig:avAbsCall}). 
However, it significantly loses accuracy for pure 
hyperbolic graphs like Hypergrid(Figure~\ref{fig:avAbsHyper}) 
for short distances.   
We note that our proposed
Hyper BFS technique shows intriguingly low 
mean error on real-world graphs: on the same range as 
Rigel on real-world graphs and better for others. 
Figure~\ref{fig:avAbsError} also shows that the other 
techniques, 10 Gromov and Steiner trees,  
achieve poor results,
particularly for node pairs at large distances. Compared to 
10-Hyper BFS and Rigel, the Gromov tree contraction 
technique results in large
errors, not only on the call graph and SantaBarbara 
Facebook
graph shown in these figures, but also on the other real-
world graphs
that we considered. This technique contracts many nodes 
into the same
high degree nodes, thereby underestimating their distances. 
For instance, on
node pairs at a distance of 13 from each other on the call
graph, it has an average absolute error of around 6.5, i.e. 
around 50\%.

\para{Short-Long Distance Approximation Accuracy.}
Here we study how different topological structures impact
the fidelity of various techniques on different distance
lengths. An important observation that follows from
Figure~\ref{fig:avAbsError},  
is that Rigel and 10-Hyper BFS result in better accuracy 
for node pairs with long
distances while Gromov performs consistently poorer for 
such
node pairs. We expect that the consistent accuracy 
of 10-Hyper BFS and Rigel, for
farther node pairs, is due to their having nodes 
with high 
betweenness centrality (as
approximated by degree centrality) at the root 
or center of their
embedding. This ensures that the distances for a large 
number of node pairs whose
shortest path passes through the core nodes are correctly 
estimated.  Interestingly enough, doing the same for 
`Gromov Tree' did not seem to improve its accuracy by
much, possibly due to its large number of nodal contractions.
For short distances instead, Rigel accuracy, in particular 
on the HyperGrid (Figure~\ref{fig:avAbsHyper}), leads to a 
mean absolute error  of 8 for the
real distance of 3, but gets more accurate for
node-pairs at larger distances.  This is clearly due
to the fact the a planar graph hardly fills a
hyperbolic space and hyperbolic embedding is a
poor fit.  
Finally, on the SquareGrid network (Figure~
\ref{fig:avAbsFlat}), which is typical of planar graphs and 
close to road networks, some techniques like Rigel and 
Steiner have very different behaviors compared to the 
previous networks. Rigel, for instance, estimates long 
distances with a very large absolute
errors. On the contrary, Steiner shows a very high error 
for distances less than 100 and then for some reason it 
dramatically improves for very long distances.
The fact that these techniques perform well on real-world 
graphs
and synthetic HyperGrid graphs but not on SquareGrid 
graphs, 
suggests that the reported success of these
techniques most likely relies on the intrinsic hyperbolicity 
of real world graphs and use of root nodes 
in the hyperbolic core, two properties not reported
heretofore. However, 10-Hyper BFS seems to cope with 
flatness well, most likely due to its BFS spanning tree
structure than its other features such as use of 
highly central nodes as root, 
which a square grid does not posses.

\begin{figure*}
\subfigure[Average Additive Error]
{
\epsfig{figure=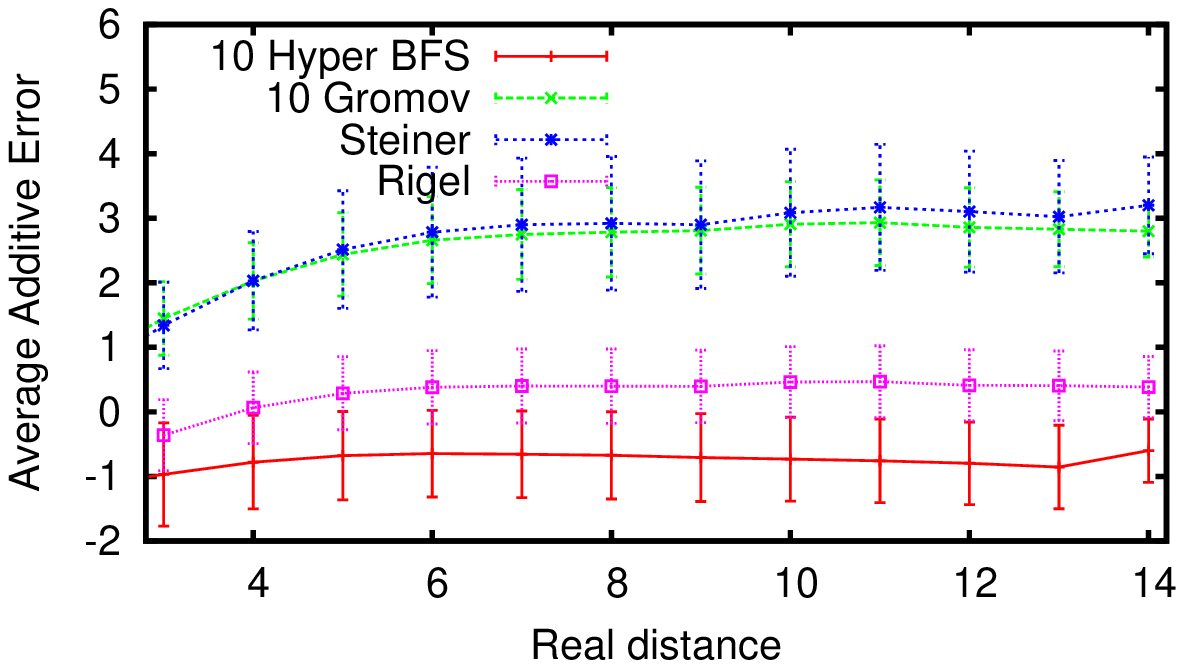, scale=0.675}
%\caption{Average additive error of various approximation techniques on the SantaBarbara Facebook graph}
\label{fig:avAddSB}
}
%\end{figure}
%\caption{Average additive error of various approximation techniques on the SantaBarbara Facebook  graph}
%\label{fig:avAddSB}
%\end{center}
%\begin{figure}
%\centering
\subfigure[Average Multiplicative Error]
{
\epsfig{figure=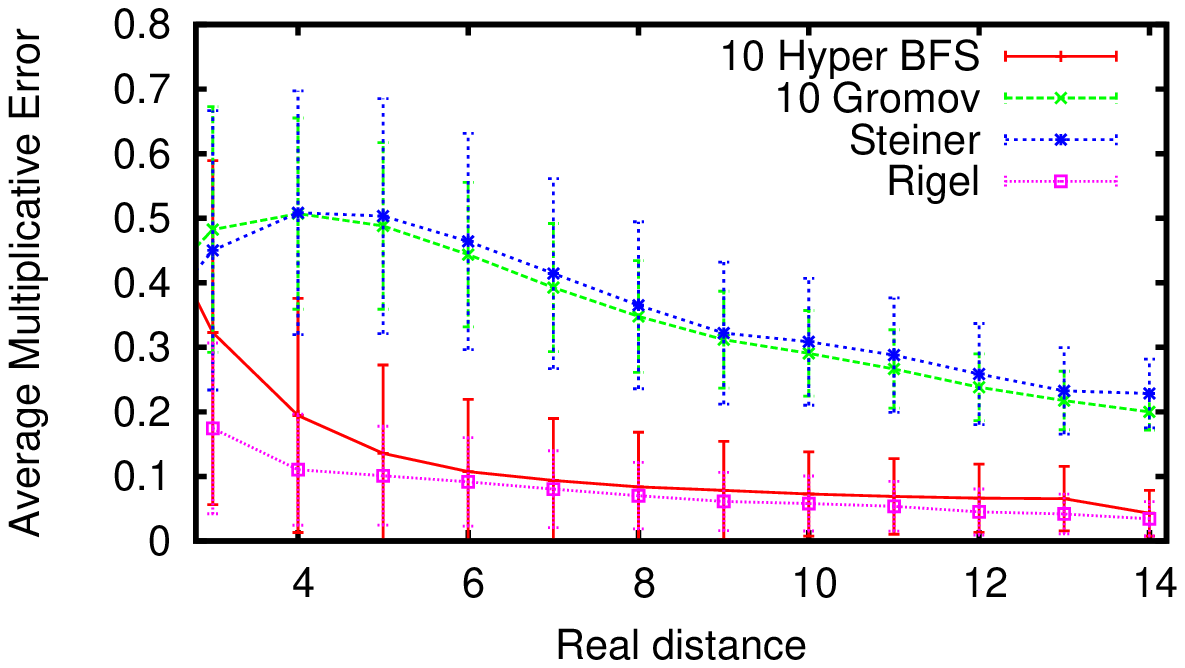, scale=0.675}
%\caption{Average multiplicative error of various approximation techniques on the SantaBarbara Facebook  graph}
\label{fig:multSB}
}
\caption{Comparing accuracy of the various approximation techniques on the SantaBarbara Facebook graph.}
\end{figure*}

\para{Fidelity Versus Theoretical Bounds.} 
Figures~\ref{fig:avAddSB} and~\ref{fig:multSB} show the additive and multiplicative errors of the various techniques
on SantaBarbara Facebook graph, to provide a quantitative comparison against the respective theoretical bounds provided by the Gromov and the Steiner constructions. 
For $\delta$-hyperbolic graphs, the theoretical worst-case bound for
the Gromov tree based approach is an additive factor of $\delta \log{n}$. For the social networks we considered, $\delta$ was a small constant,
but owing to large sizes of the graph, this theoretical bound of
$\delta \log{n}$ is more than the
diameter of the graph. However, in Figures~\ref{fig:avAddSB} 
we observe that the average additive error of
Hyper BFS is much smaller than the theoretical bound, and
the Gromov tree contraction is evidently the most erroneous compared 
to other techniques considered.
%Figures~\ref{fig:avAddSB} and~\ref{fig:multSB} show the additive and multiplicative errors of various techniques on SantaBarbara Facebook graph. 
Similarly, the theoretical guarantee for $\varepsilon$-approximate
Steiner tree approach is quite large for the studied graphs, in Figure~\ref{fig:multSB} we observe a
comparatively small multiplicative error of less than 0.6 for most
node pairs on this graph. This provides further evidence that the
worst-case bounds for the accuracy of these oracles are
pessimistic and the average error on real-world graphs is actually
much less. Rigel and 10-Hyper BFS approaches are 
significantly better the 
theoretical bounds, by exhibiting both a smaller additive and 
multiplicative errors 
compared to the Gromov and Steiner techniques.
These results confirm that the shortest paths in these graphs are
well-captured by `backbone trees', which is in line with
our understanding of the intrinsic hyperbolicity of 
the real-world graphs. 

\para{Exploiting General Hyperbolic Topologies.}
To reduce the dependence of Hyper BFS on the correlation between
degree and betweenness (centrality) of nodes and, thereby, 
to make this technique
more accurate on a wider range of graph classes, 
we consider alternative
strategies for selecting the root nodes for our tree oracle. We identify
a node with high closeness\footnote{A high closeness central node is a node with low shortest path distances to all other nodes in the graph.} centrality by selecting a random node $u$,
finding the node $m_u$ that is at maximum distance from $u$, selecting
a node $m_{m_u}$ at maximum distance from $m_u$ and returning the node
in the middle of a path between $m_u$ and $m_{m_u}$ in the
graph. We note that similar techniques have been used to
approximate diameter of large graphs (see e.g.,~\cite{crescenzi10}).

We aim for a technique that would work well on a general graph topologies, 
particularly real-world social interaction graphs. 
To this end, we
consider a diverse seeding strategy, in which the 
first seed node
is random, the second seed is selected at maximum distance from first
seed, the third seed is selected based on closeness centrality (as
described before) and the remaining seeds are selected from among high
degree nodes. We found that this {\em diverse}
seeding strategy performs well on various graph classes, and in
particular on HyperGrid, its accuracy is close to the best results from degree
based root selection and closeness based root selection methods
(cf. Figure~\ref{fig:avAbsBFSHyper}).

\begin{figure}
\centering
\includegraphics[scale=0.25, angle=270]{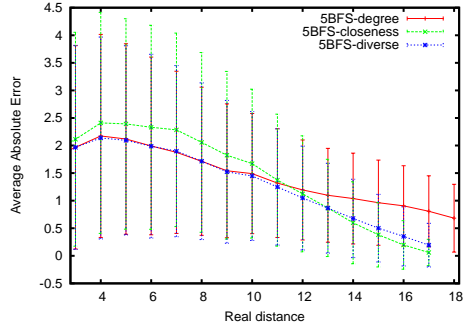}
\caption{Average absolute error of various Hyper BFS approaches with
  diverse seeding on
the HyperGrid network.}
\label{fig:avAbsBFSHyper}
\vspace{-0.1in}
\end{figure}

\subsection{Role of Hyperbolic Core}
The above benchmarking tests show that Hyper BFS is 
at least as fast as other tree-based oracles and gives 
distance estimates comparable to the best.
We recall that Hyper BFS constructs a spanning tree of 
the graph based on the ordering of nodes
by their centrality, with a root in the hyperbolic
core of the graph.  What aspect of Hyper BFS is
key to it performance? Here we argue that
the selection of the root node of the BFS spanning tree 
in the hyperbolic core
is likely the most critical element in 
making Hyper BFS a strong distance oracle for real-world 
networks.  We show this by two sets of experiments; The
first set keeps the root node in the hyperbolic core but
changes the rest of the spanning BFS tree,
and the second set changes the root but keeps 
everything else the same.
%%%%%%%HHHHHHHHHHHHHHHHHHHHHHHHHHHH
%% \epsfig{figure=avAddErrorSantaBarbara_final.eps, scale=0.7}

\begin{figure*}[!htb]
\subfigure[DBLP]
{
\epsfig{figure=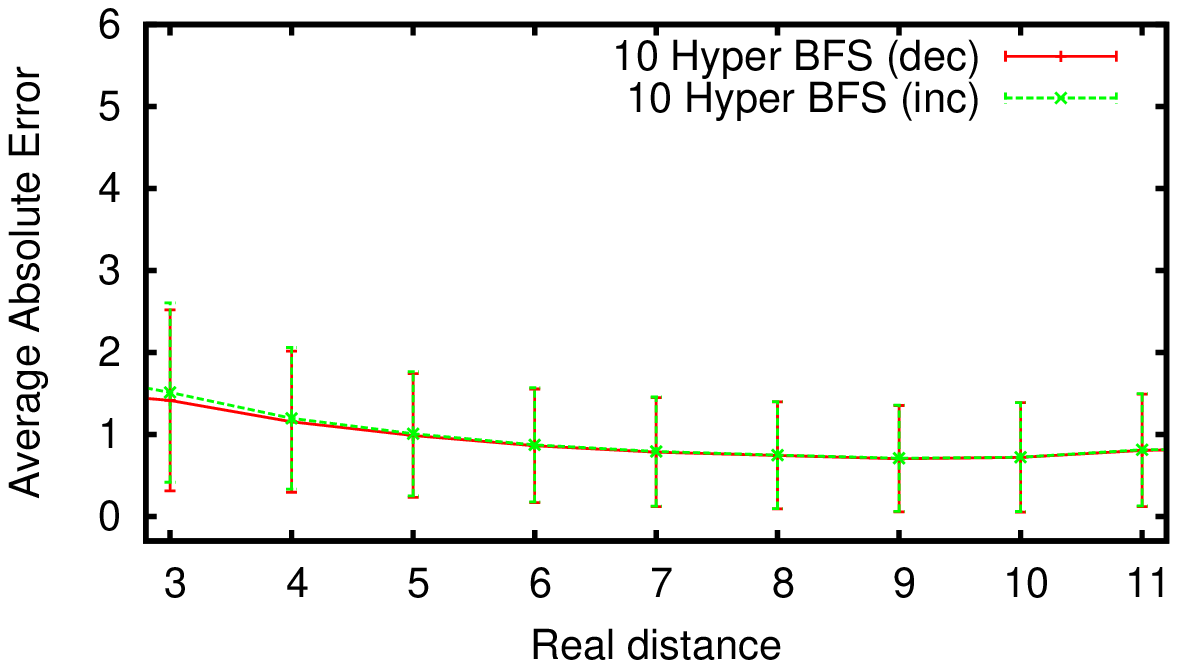,scale=0.675}
\label{fig:avAbsDBLP-dec}
}
%\subfigure[Web Berk-Stan]
%{
%%\epsfig{figure=avAbsErrorWebBerkStan_inc_dec.eps,width=3.15in}
%\epsfig{figure=avAbsErrorWebBerkStan_inc_dec.eps,scale=0.675}
%\label{fig:avAbsWebBerkStan-dec}
%}
\subfigure[Santa Barbara, Facebook]
{
\epsfig{figure=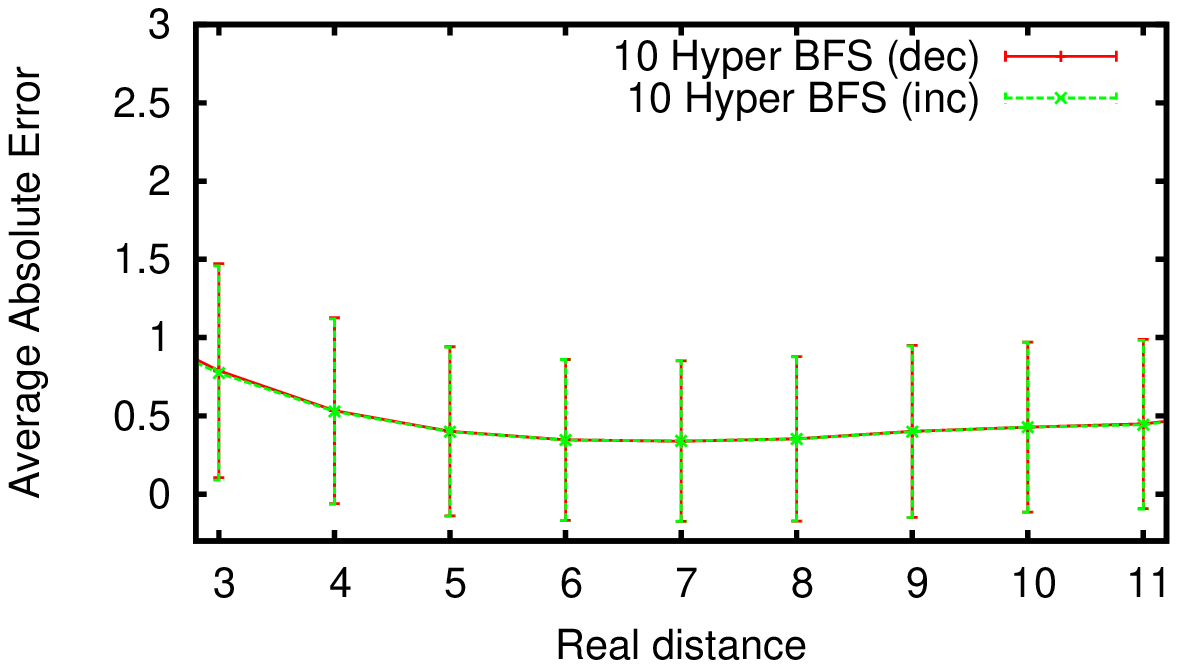,scale=0.675}
\label{fig:avAbsSB-dec}
}
\subfigure[DBLP]
{
\epsfig{figure=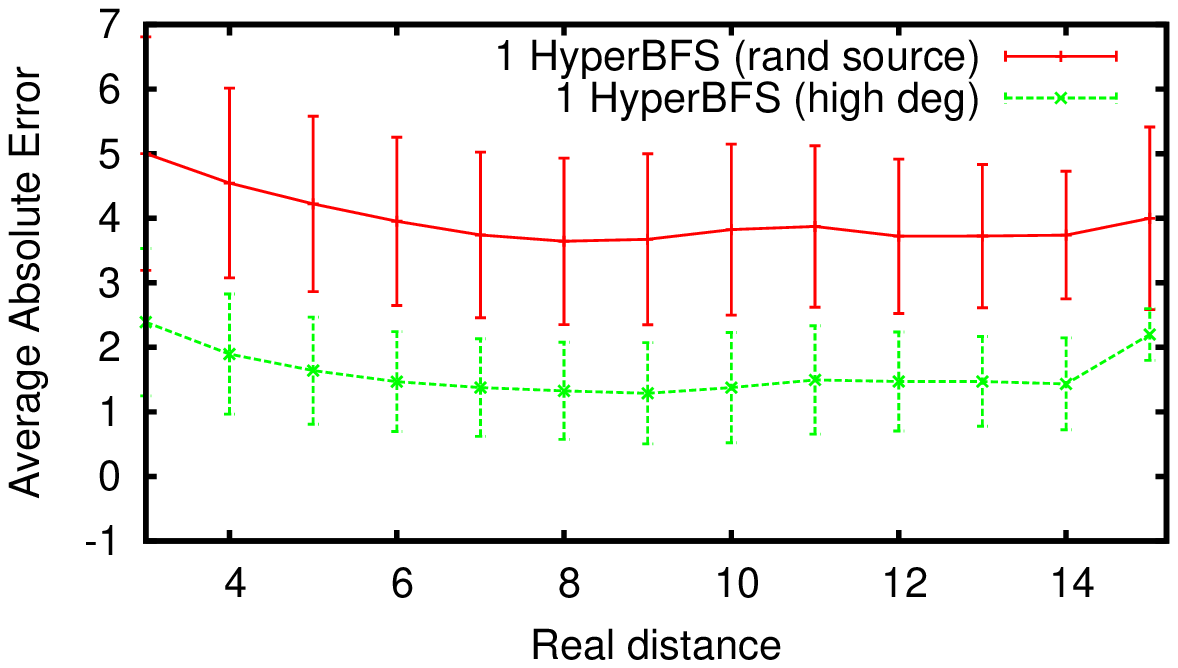,scale=0.675}
\label{fig:avAbsDBLP-randRoot}
}
\subfigure[Santa Barbara, Facebook]
{
\epsfig{figure=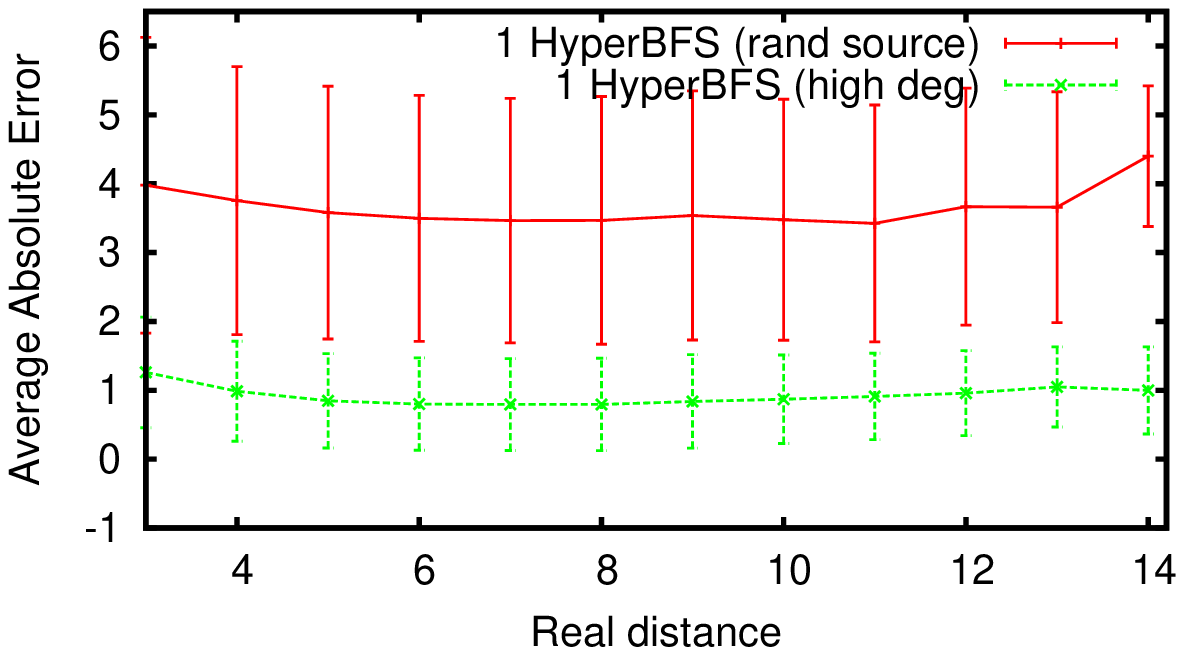,scale=0.675}
\label{fig:avAbsSB-randRoot}
}
%\subfigure[Call Graph II]
%{
%%\includegraphics[trim=0cm 0cm 0cm 0cm,clip=true,scale=0.7]{avAbsErrorSKT_final_inc_dec.eps}
%%\epsfig{figure=avAbsErrorSKT_final_inc_dec.eps,width=3.15in}
%\epsfig{figure=avAbsErrorSKT_final_inc_dec.eps,scale=0.675}
%\label{fig:avAbsSKT-dec}
%}
\caption{Top (a)-(b): Comparison of Hyper BFS and Hyper BFS(inc) 
with increasing nodal degree, on two representative 
networks both with highest centrality node as BFS root.  
Bottom (c)-(d): Two instances of 
Hyper BFS on the same two networks, one with highest centrality
node as root (green) and the other with a random node as root (red).}
\label{fig:antiHyper}
\end{figure*}

Figures~\ref{fig:avAbsDBLP-dec}-\ref{fig:avAbsSB-dec} 
show comparisons on 
two sets of benchmark real-world networks between 
Hyper BFS and another implementation of Algorithm 1,
which we label Hyper BFS(inc). 
In this implementation of Hyper BFS,
the root node is still selected from the
hyperbolic core (i.e., has the highest centrality) but 
the ordering of the subsequent nodes in the BFS 
spanning tree is in reverse of Hyper BFS, that is,
in order of increasing degrees.  
As it can be seen, the results are almost identical 
in terms of average distortion compared to Hyper BFS.  
However,
when we start the Hyper BFS tree with a random root node
as shown in 
Figures~\ref{fig:avAbsDBLP-randRoot}-\ref{fig:avAbsSB-randRoot},
then we see a factor 2-4 increase in the 
absolute error
compared to Hyper BFS where the root has highest
centrality. This result, perhaps
initially surprising, is consistent with our
understanding of the structure of real-world graphs,
since the key property of $\delta$-hyperbolic
networks is the existence of a core whose centrality is
$O(n^2)$, thus ensuring that a large fraction of all
shortest paths automatically traverse this small set
and for which distance error is zero for that many
node pairs.

\subsection{Scalability of Hyper BFS}
The Call Graph II has $\sim$50 million nodes 
and $\sim$300 million edges.  
But as it can be seen from
Figure~\ref{fig:avAbsErrorSKTHyper} and Table~\ref{table:scalabilityHyerBFS_SKT}, the total
run time for the Hyper BFS on this network was under 
one minutes (under 10 minute for 10 Hyper BFS), 
including one million node-pair distance 
queries that were completed in 25 seconds, all on a 
2.4 GHz Intel(R) Xeon(R) processor with 190 GB of RAM.
We do not know of faster 
distance approximation techniques with error 
as small as shown in 
Figure~\ref{fig:avAbsErrorSKTHyper}.
In our implementation of Hyper BFS, we used a 
standard tree labeling scheme whereby we
store the list of parent nodes of each node to the
root, $n k$ indices for $k$ Hyper BFS trees on $n$ nodes.
Since the height of the Hyper BFS tree is $O(\log{n})$,
a query has $O(k \log{n})$ complexity.
We did not run the other three oracles on
this data set as competitive implementations 
for a graph of this size required 
substantial optimization of the code and investment of
time.  We expect Rigel to require several hours for 
computation of its embedding phase but from there 
it is likely competitive with Hyper BFS in completing one 
million queries in 20-30 seconds and based on
our observations from other sample real networks,
we expect similar or marginally better accuracy.
%This makes the overall efficiency of Hyper BFS
%over 100 times (faster) with competitive (small) error.

\begin{figure}
\centering
\includegraphics[scale=0.65]{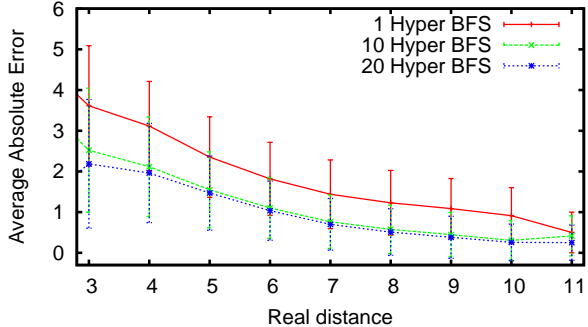}
\caption{Hyper BFS accuracy on a 50M node call-graph with different number of trees, i.e. from 1 to 20.}
\label{fig:avAbsErrorSKTHyper}
\vspace{-0.1in}
\end{figure}

\begin{table}[t]
\begin{center}
\caption{\small Computational Time of Hyper BFS on Call Graph II.}
\vspace{0.1in}
%\begin{footnotesize}
\begin{tabular}{|c|c|c|}
\hline
  Loading Graph& Hyper BFS Tree & 1M Queries \\
 \hline
 250 sec& 50 sec& 25 sec\\
 \hline
\end{tabular}
\label{table:scalabilityHyerBFS_SKT}
%\end{footnotesize}
\end{center}
\vspace{-0.2in}
\end{table}

\subsection{Bounding the Hyper BFS Error}
\label{sec:approx_range}
\fixme{Since the Hyper BFS is a spanning tree of the graph, the distance $d_H(x,y)$ in
the Hyper BFS tree is an upper bound on the distance $d_G(x,y)$ in the
graph for any node pair $(x,y)$. On the
other hand, a Gromov tree is a contraction of the input graph and the
distance in the Gromov tree $d_{GT}(x,y)$ is a lower bound on the
graph distance $d_G(x,y)$ (cf. Lemma~\ref{lemma:gromov_contraction}). Thus, we can use these
approaches together to get an approximation range $[l(x,y),
u(x,y)]$ with guaranteed upper $(u(x,y))$ and lower $(l(x,y))$ bound for the actual
distance $d_G(x,y)$ in the graph, i.e., $l(x,y) \leq d_G(x,y) \leq
u(x,y)$. Since $d_G(x,y) \leq d_{GT}(x,y) + 2 \delta \log{n}$, we can
use $u(x,y) = \min \{d_H(x,y), d_{GT}(x,y) + 2 \delta \log{n} \}$ and
$l(x,y) = d_{GT}(x,y)$. By the above definition of the upper and lower
bound, it follows that the width of the approximation range
$u(x,y) - l(x,y)$ is less than equal to $2 \delta \log{n}$.}

\fixme{We can reduce the width of this approximation range at the cost of
increasing the precomputation time and the storage space by running
multiple runs of Hyper BFS tree and the Gromov Tree from different root
nodes. Thus, the new lower bound is the maximum over all Gromov Trees.
The new upper bound is the minimum over all the hyper BFS trees and the
Gromov trees with additive error bound,  i.e., $u(x,y) =
\min\{\min_H\{d_H(x,y)\}, \min_{GT}\{d_{GT}(x,y) + 2 \delta \log{n}\}\}$. 
Note that both in the Hyper BFS tree and in the Gromov tree, the
distance from the root $r$ to any node $x \in V$ is exact, i.e.,
$d_H(r,x) = d_{GT}(x,y) = d_G(x,y)$. Thus, if we use $n$ different
hyper BFS trees with different root nodes, the minimum over them will
yield the exact graph distance. Similarly, if we use $n$ different
Gromov trees with different root nodes, the maximum over them will
give the exact graph distance, resulting in a approximation range
width of zero. However, this extreme point of the solution space
requires $O(n^2)$ preprocessing time and $O(n^2)$ storage space.}

\begin{figure}
\centering
\includegraphics[scale=0.65]{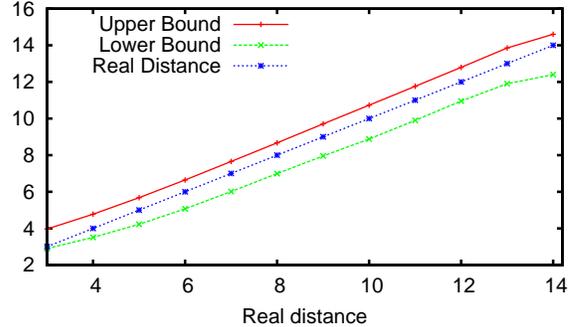}
\caption{Upper and lower bounds and the width of the approximation
  range by using 10-Hyper BFS tree and 20 Gromov trees on the
  SantaBarbara Facebook graph.}
\label{fig:up_low_sb}
\vspace{-0.05in}
\end{figure}

\fixme{The accuracy of the Hyper BFS approach is seen to be much
smaller than the $2 \delta \log{n}$ theoretical bound for Gromov
trees, on the real-world graphs that we tested. This implies that we
can obtain a fairly small approximation range by combining the
Hyper BFS and the Gromov tree approach. In fact, as
Figure~\ref{fig:up_low_sb} shows, the upper and lower bounds on
SantaBarbara Facebook graph are quite
close and the resultant range with just 10-Hyper BFS trees and 20
Gromov trees is quite small.}

\section{Conclusion}
\label{sec:conclusion}
In this work, we construct a novel distance
oracle that is based on the intrinsic hyperbolicity 
of the graph and specifically leverages its
hyperbolic core.  Rooting a breadth-first-search spanning 
tree within the 
hyperbolic core of the graph ensures that a fixed
fraction of all
shortest paths have no distance estimation error, 
and the remaining ones have small error.  Furthermore,
the smaller the hyperbolic
constant $\delta$ of the graph, the better such a tree
approximation is expected to be.  
%the smaller the graph hyperbolic 
%core (constant), the 
%better the distance approximation 
%via spanning breadth-first search trees rooted
%in the core.  

We tested an implementation of this theory-based framework 
on a dozen real and synthetic large graphs, from 10s of
thousand to 10s of millions of nodes, and found that in 
practice it leads to surprisingly fast and accurate results, 
significantly better than 
anticipated by existing
theoretical error bounds that have been proven for 
alternative tree approximations of the graph.
This theoretical framework, especially the existence of
the hyperbolic core of the graph, together with evidence 
from our experiments
also suggest that the success of prior heuristic distance
oracles may also be due to the same reasons: 
1) the underlying
hyperbolicity of the real-world graphs and 
2) good correlation between nodal 
degrees and their
betweenness centrality in many real-world graphs. 
In particular, $\delta$-hyperbolicity implies
that in these graphs, shortest paths can be well-
approximated by appropriately rooted 
`backbone spanning trees'.

Another interesting observations is that a 
simple approach of
computing a few BFS trees from high centrality/degree seed nodes
provides quite accurate results. 
To improve this technique further,
we considered two strategies: 1) 
selecting a diverse set of seed nodes for growing 
BFS trees and 2) to grow the tree by 
expanding along high 
degree nodes first. We found that our first strategy 
helped make the
technique more robust -- 
it helped the 20 BFS provide good accuracy
even on hypergrid graphs where there was no 
degree distribution as proxies for betweenness centrality
of nodes in the BFS expansion. 
The second strategy helped on 
some synthetic graphs, such as hypergrids,
but it did not improve distance approximation on 
real-world networks: it makes a big difference to
root the spanning tree approximation at
a vertex in the hyperbolic core.
We note additionally that
compared to the best performing distance approximation 
oracles, 
such as Rigel, our approach has a very low computing 
cost for creating the oracle and
thus lends itself well to settings where the graph 
connectivity changes frequently, such as in
large scale dynamic graphs.

We contend that further theoretical understanding of 
the geometry of
$\delta$-hyperbolic graphs can lead to even better 
distance oracles for real-world graphs. For example,
it would be helpful to know how large the hyperbolic 
core of a $\delta$-hyperbolic graph can be.
Another interesting open problem is to find alternative 
proxies for betweenness centrality in real-world 
graphs which, as with nodal degrees, is easy to compute.

\para{Acknowledgment.}  
The work of Iraj Saniee and Sean Kennedy was supported 
by the AFOSR grant no. FA9550-11-1-0278 and the
NIST grant no. 60NANB10D128.
%
%\bibliographystyle{abbrv}
%\bibliography{biblo}

\end{document}